\documentclass[11pt,reqno]{article}
\usepackage{amsmath,amssymb,amsthm}
\usepackage{color,url}
\usepackage{graphicx}
\newtheorem{theorem}{Theorem}[section]
\newtheorem{proposition}[theorem]{Proposition}

\newtheorem{definition}[theorem]{Definition}
\newtheorem{lemma}[theorem]{Lemma}
\newtheorem{example}[theorem]{Example}
\newtheorem{remark}[theorem]{Remark}

\numberwithin{equation}{section}

\usepackage[dvipdfmx]{hyperref}

\def\bR{\mathbb{R}}
\def\bC{\mathbb{C}}
\def\bN{\mathbb{N}}

\def\bM{\mathbb{M}}
\def\bT{\mathbb{T}}
\def\M{\mathcal{M}}
\def\N{\mathcal{N}}
\def\cH{\mathcal{H}}
\def\cP{\mathcal{P}}
\def\cE{\mathcal{E}}
\def\B{\mathcal{B}}
\def\cA{\mathcal{A}}
\def\cC{\mathcal{C}}
\def\cI{\mathcal{I}}
\def\cS{\mathcal{S}}
\def\cD{\mathcal{D}}
\def\cT{\mathcal{T}}
\def\fM{\mathfrak{M}}

\def\<{\langle}
\def\>{\rangle}
\def\ffi{\varphi}
\def\eps{\varepsilon}
\def\ffi{\varphi}
\def\tr{\mathrm{tr}}
\def\Tr{\mathrm{Tr}}
\def\HS{\mathrm{HS}}

\def\meas{\mathrm{meas}}
\def\test{\mathrm{test}}
\def\1{\mathbf{1}}

\def\nw{^{*}}

\def\hil{\mathcal{H}}

\def\scli{\underline{sc}}
\def\scls{\overline{sc}}

\newcommand{\norm}[1]{\left\|#1\right\|}
\newcommand{\ds}{\mbox{ }\mbox{ }}

\newcommand{\vertleq}{\rotatebox{90}{$\,\ge$}}

\newcommand{\derleft}[1]{\partial^{-} #1}
\newcommand{\derright}[1]{\partial^{+} #1}

\begin{document}

\centerline{\LARGE Quantum R\'enyi divergences and}
\medskip
\centerline{\LARGE the strong converse exponent of state discrimination}
\medskip
\centerline{\LARGE in operator algebras}

\bigskip
\bigskip
\centerline{\large
Fumio Hiai\footnote{{\it E-mail:} hiai.fumio@gmail.com}
and Mil\'an Mosonyi\footnote{{\it E-mail:} milan.mosonyi@gmail.com}$^{,3}$}

\medskip
\begin{center}
$^1$\,Graduate School of Information Sciences, Tohoku University, \\
Aoba-ku, Sendai 980-8579, Japan
\end{center}

\begin{center}
$^2$\,MTA-BME Lend\"ulet Quantum Information Theory Research Group \\
$^3$\,
Department of Analysis, Institute of Mathematics, Budapest University of Technology and
Economics, M\H uegyetem rkp. 3., H-1111 Budapest, Hungary
\end{center}

\medskip

\begin{abstract}
The sandwiched R\'enyi $\alpha$-divergences of two finite-dimensional quantum states play a 
distinguished role among the many quantum versions of R\'enyi divergences as 
the tight quantifiers of the trade-off between the two error probabilities in the strong converse
domain of state discrimination. In this paper we show the same for the sandwiched R\'enyi
divergences of two normal states on an injective von Neumann algebra, thereby establishing the
operational significance of these quantities.  Moreover, we show that in this setting, again similarly
to the finite-dimensional case, the sandwiched R\'enyi divergences coincide with the regularized
measured R\'enyi divergences, another distinctive feature of the former quantities. Our main tool
is an approximation theorem (martingale convergence) for the sandwiched R\'enyi divergences,
which may be used for the extension of various further results from the finite-dimensional to the
von Neumann algebra setting.

We also initiate the study of the sandwiched R\'enyi divergences of pairs of states on a
$C^*$-algebra,  and show that the above operational interpretation, as well as the equality to the
regularized measured R\'enyi divergence, holds more generally for pairs of states on a nuclear
$C^*$-algebra.



\end{abstract}

{\baselineskip10pt
\tableofcontents
}

\section{Introduction}\label{Sec-1}

R\'enyi's $\alpha$-divergences \cite{Renyi} give a one-parameter family of pseudo-distances on
probability measures, which  play a central role in information theory as quantifiers of the trade-off 
between the relevant operational quantities in many information theoretic tasks; see, e.g., 
\cite{Csiszar}. In quantum information theory, the non-commutativity of density operators
allows infinitely many different extensions of the classical R\'enyi divergences to pairs of 
finite-dimensional quantum states; among others, the standard (or Petz-type) R\'enyi divergences
\cite{Petz}, the sandwiched R\'enyi divergences \cite{MDSFT,WWY}, their common generalization,
the R\'enyi $(\alpha,z)$-divergences \cite{AD,JOPP}, or the maximal (or geometric) R\'enyi
divergences \cite{Matsumoto,PetzRuskai}. Many of these notions have also been extended to pairs
of density operators on infinite-dimensional Hilbert spaces, or even to pairs of normal states on
a von Neumann algebra \cite{BST,Hi3,Hi4,Hi,Je1,Je2,Mo,Pe1}. The study of these quantities has
been motivating extensive research in the fields of matrix analysis and operator algebras.

On the other hand, in quantum information theory the relevant problem is to identify the
quantum R\'enyi divergences with operational significance, i.e., those which appear as
natural quantifiers of the trade-off relations between the quantities describing a problem, like
error probabilities or coding rates. This has been established for the standard R\'enyi divergences
with parameter values $\alpha\in(0,1)$ in the context of binary state discrimination
(hypothesis testing) of finite-dimensional quantum states in a series of works
\cite{Aud,Hayashi,HT,HMO,MH,Nagaoka,NSz}, and was also extended to the von Neumann algebra
setting in \cite{JOPS}. Complementary to this, the sandwiched R\'enyi divergences were shown to
have operational significance for the parameter values $\alpha>1$ in the strong converse problem
of binary state discrimination of finite-dimensional quantum states \cite{HT,MO1,MO4}, and
in the strong converse problem of classical-quantum channel coding \cite{MO2,MO3}.
Apart from the standard and the sandwiched R\'enyi $\alpha$-divergences mentioned above,
no other quantum R\'enyi divergence has been shown to have a direct operational interpretation
so far.

It is also a problem of central importance how much the distinguishability of two states,
as measured by a quantum divergence, changes under quantum operations, in particular, under
quantum measurements. While there is no known explicit formula for the optimal
post-measurement R\'enyi $\alpha$-divergence (called the measured R\'enyi divergence),
it is known to be strictly smaller than the standard R\'enyi $\alpha$-divergence \cite{BFT,HM2}.
Interestingly, if the measured R\'enyi $\alpha$-divergence is evaluated on several copies of 
the states, and normalized by the number of copies, then the asymptotic limit of these quantities,
called the regularized measured R\'enyi $\alpha$-divergence, turns out to coincide with
the sandwiched R\'enyi $\alpha$-divergence for all $\alpha\in[1/2,+\infty)$. This is another
feature distinguishing the sandwiched R\'enyi divergences among the multitude of different
quantum R\'enyi divergences. 

The proof of the hypothesis testing interpretation of the sandwiched R\'enyi divergences in
\cite{MO1} goes via replacing the quantum i.i.d.~problem with a non-i.i.d.~classical hypothesis
testing problem by block-diagonalizing (pinching) large tensor powers of the first state 
by the spectral projections of the same tensor powers of the second state. It can be shown that
the resulting classical problem has the same optimal error asymptotics as the original quantum
problem, by using the pinching inequality \cite{H:pinching} and the fact that the number of distinct
eigenvalues of the $n$th tensor power of a density operator grows only polynomially in $n$,
even though the dimension of the underlying Hilbert space grows exponentially. The same
technique can be used to show the equality of the regularized measured R\'enyi divergences and
the sandwiched R\'enyi divergences \cite{MO1}. While the pinching technique is very simple, it is
also very powerful (see, e.g., \cite{SBT} for further applications), but its applicability is obviously
limited to the finite-dimensional case. Hence, even though the sandwiched R\'enyi divergences
have been defined for pairs of normal states on a von Neumann algebra already some time ago
\cite{BST,Je1,Je2}, it has been an open problem (as proposed in \cite{BST}) whether they have
an operational significance similar to the finite-dimensional case. This has been confirmed
very recently in \cite{Mo} in the simplest case where the von Neumann algebra is the space of
bounded operators on an infinite-dimensional Hilbert space, using a finite-dimensional
approximation technique, in particular, by showing that the R\'enyi divergences of the restrictions
of the states to finite-dimensional subspaces converge to the R\'enyi divergences of the original
states as the subspaces increase to the whole space. 

In this paper we extend the above results about the sandwiched R\'enyi divergences to
considerably more general settings, including injective, i.e., approximately finite-dimensional (AFD)
von Neumann algebras. Our main tool is again finite-dimensional approximation. More generally,
we show in Theorem \ref{T-3.1} that the sandwiched R\'enyi divergences have the martingale
convergence property for every $\alpha\in[1/2,+\infty)\setminus\{1\}$, i.e., if an increasing net of
von Neumann subalgebras generates the whole algebra then the sandwiched R\'enyi
$\alpha$-divergences of the restrictions of two states converge to the sandwiched R\'enyi
$\alpha$-divergence of the original states. The proof is based on variational representations of
the sandwiched R\'enyi divergences \cite{Hi,Je2}, and the martingale convergence of generalized
conditional expectations from \cite{HT1}. Using this result, we show in Theorem \ref{T-3.7} that
the strong converse exponents of discriminating two normal states of an injective algebra are
equal to their Hoeffding anti-divergences, analogously to the finite-dimensional case \cite{MO1}
and the case where the algebra is $\B(\hil)$ \cite{Mo}. Based on this result, in Theorem
\ref{T-3.11} we give a direct operational representation of the sandwiched R\'enyi
$\alpha$-divergences as generalized cutoff rates, following Csisz\'ar's approach \cite{Csiszar}.
Finally, using again the martingale convergence property, in Proposition \ref{P-3.13} we show
that for any $\alpha\in[1/2,+\infty)\setminus\{1\}$, the sandwiched R\'enyi $\alpha$-divergence
of two states on an injective algebra coincides with their regularized measured R\'enyi
$\alpha$-divergence. 

Moreover, we also initiate the study of the sandwiched R\'enyi divergences for states of a 
$C^*$-algebra. In Theorem \ref{T-4.3} we show that for any two states on a $C^*$-algebra, and
any representation of the algebra that admits normal extensions of the states to the 
generated von Neumann algebra, the sandwiched R\'enyi $\alpha$-divergence of the extensions
is independent of the specific representation for any $\alpha\in[1/2,+\infty)$, and hence it 
gives a well-defined notion of sandwiched R\'enyi $\alpha$-divergence of the original states. 
We also show the same statement for the standard $\alpha$-divergence and every 
$\alpha\in[0,+\infty)\setminus\{1\}$. In Proposition \ref{P-4.5} we establish the basic properties
of these extensions: joint lower semi-continuity, monotonicity under composition with unital
positive maps (Schwarz maps in the case of the standard R\'enyi divergences), the inequality
between the sandwiched and the standard  R\'enyi divergences, and the martingale convergence
property for both. In Theorem \ref{T-4.12} we show that the sandwiched R\'enyi divergences on
nuclear $C^*$-algebras have the same operational interpretation as in the case of injective
von Neumann algebras, i.e., we show the equality of the strong converse exponents and the 
Hoeffding anti-divergences, from which the generalized cutoff rate representation also follows
immediately. Finally, in Proposition \ref{P-4.14} we show that the sandwiched R\'enyi divergences
coincide with the regularized measured R\'enyi divergences on nuclear $C^*$-algebras.

The main text is accompanied by eight appendices. 
In Appendices \ref{sec:relmodop}--\ref{sec:injective}, we give brief overviews
of the notions and concepts in von Neumann algebra theory that we use in the paper.
Our general reference for this part is \cite{Hi6}.
In Appendix \ref{sec:finitedim}, we fill a gap
in the proof of the equality of the strong converse exponents and the Hoeffding anti-divergences
in the finite-dimensional case \cite{MO1}, and show the same equality for a slightly modified
definition of the strong converse exponent. Appendix \ref{sec:boundary} contains a simple
observation about the boundary values of convex functions, needed in the proof of
Theorem \ref{T-3.7}. Finally, Appendix \ref{sec:proof} contains the rather technical proof of
Theorem \ref{T-4.3}.

\section{Sandwiched and standard R\'enyi divergences}\label{Sec-2}

In this section we briefly review the notions of sandwiched and standard R\'enyi divergences in 
von Neumann algebras. For a more detailed exposition, see, e.g., \cite{Hi}.
We refer the reader to \cite{Hi6} for the necessary concepts in von Neumann algebra theory, some
of which we will also briefly explain here and in the Appendices, for the convenience of the reader.

The notion of sandwiched R\'enyi divergences with parameter 
$\alpha\in[1/2,+\infty)\setminus\{1\}$, introduced first in \cite{MDSFT,WWY} for pairs of
finite-dimensional density operators, was generalized by Berta--Scholz--Tomamichel \cite{BST}
and Jen\v cov\'a \cite{Je1,Je2} to the general von Neumann algebra setting. The definitions are
different from each other between the three papers \cite{BST,Je1,Je2} but their equivalence was
proved in \cite{Je1,Je2} (also \cite[Sec.~3.3]{Hi}). Here we work with the definition in \cite{Je1}
based on Kosaki's interpolation $L^p$-spaces. 

For a von Neumann algebra $\M$, let $\M_*^+$ denote the set of positive normal functionals on
$\M$. The identity of $\M$ is denoted by $\1$. If $\M=\B(\hil)$ is the von Neumann algebra of
all bounded operators on a  finite-dimensional Hilbert space, then any $\psi\in\M_*^+$ can be
represented by a positive semi-definite operator $\hat\psi\in\B(\hil)_+$ such that
$\psi(x)=\Tr (x\hat\psi)$ for any  $x\in\M$. The \emph{sandwiched R\'enyi $\alpha$-divergence}
of $\rho,\sigma\in\M_*^+$ can then be defined for any $\alpha\in(0,1)\cup(1,+\infty)$ as
\cite{MDSFT,WWY} 
\begin{align}\label{F-2.1}
D_{\alpha}\nw(\rho\|\sigma):=\frac{1}{\alpha-1}\log Q_{\alpha}\nw(\rho\|\sigma),
\end{align}
where 
\begin{align}\label{F-2.2}
Q_{\alpha}\nw(\rho\|\sigma):=
\begin{cases}
\Tr\Bigl(\hat\sigma^{\frac{1-\alpha}{2\alpha}}
\hat\rho\hat\sigma^{\frac{1-\alpha}{2\alpha}}\Bigr)^{\alpha},&
\text{if}\ s(\rho)\le s(\sigma)\ \text{or}\ \alpha\in(0,1),\\
+\infty,&\text{otherwise}.
\end{cases}
\end{align}
Here, $s(\rho)$ is the smallest projection $p\in\M$ such that $\rho(p)=1$ (the
\emph{support projection} of $\rho$), and $s(\sigma)$ is defined similarly.
Real powers of a positive semi-definite operator $A\in\B(\hil)_+$
are defined as $A^x:=\sum_{\lambda>0}\lambda^{x}P_{\lambda}$, 
$x\in\bR$, where 
$P_{\lambda}$ is the spectral projection of $A$ corresponding to 
$\{\lambda\}\subseteq\bR$. The logarithm can be taken in any base that is larger than $1$, and
it is extended to $[0,+\infty]$ by $\log 0:=-\infty$, $\log(+\infty):=+\infty$.

\begin{remark}\label{R-2.1}\rm
It is customary to define $D_\alpha^*$ with a normalization like
\[
D_\alpha^*(\rho\|\sigma)={1\over\alpha-1}\log{Q_\alpha^*(\rho\|\sigma)\over\rho(\1)}
\]
(or with restricting $\rho$ to states) but we use $D_\alpha^*$ without this normalization because
that is the natural choice for the study of the strong converse exponent in  Section \ref{Sec-3.2};
see, e.g., \eqref{F-3.7}.
\end{remark}

In the case of a general von Neumann algebra $\M$, 
there need not be a trace functional on $\M$, and a useful representation 
of states as operators is not at all straightforward. To this end,
one may use Haagerup's construction (see Appendix \ref{Sec-B} for details) to obtain a larger
von Neumann algebra $\N$ with a faithful normal semifinite trace $\tau$ on it, with the
corresponding *-algebra $\widetilde\N$ of $\tau$-measurable operators affiliated with $\N$,
and Banach spaces $L^p(\M)\subseteq\widetilde \N$, with corresponding norm
$\norm{\cdot}_p$, $p\in[1,+\infty]$, such that 
\begin{itemize}
\item
$L^{\infty}(M)$ is identical to the von Neumann algebra $\M$;
\item \label{ordeiso}
there exists an order isomorphic linear bijection $\psi\mapsto h_{\psi}$
from $\M_*$ onto $L^1(\M)$;
\item
for every $\psi\in\M_*^+$ and every $p\in[1,+\infty)$, $h_\psi^{1/p}\in L^p(\M)_+$
($=L^p(\M)\cap\widetilde\N_+$), where $h_\psi^{1/p}$ is defined via standard functional
calculus;
\item
for every $p,q\in[1,+\infty]$ with $1/p+1/q=1$, and every 
$a\in L^p(\M)$, $b\in L^q(\M)$, $ab\in L^1(\M)$.
\end{itemize}
Moreover, the order isomorphism above defines a functional $\tr$ on $L^1(\M)$ as
$\tr\,h_{\psi}:=\psi(\1)$, and for every $\psi\in\M_*$ and $x\in\M$,
\begin{align*}
\psi(x)=\tr (xh_{\psi}),
\end{align*}
in complete analogy with the finite-dimensional case. For any $\sigma\in\M_*^+$, 
\emph{Kosaki's (symmetric) interpolation $L^p$-spaces} $L^p(\M,\sigma)$ with respect to
$\sigma$ for $p\in[1,+\infty]$ with $1/p+1/q=1$ are defined as
\begin{align*}
&L^p(\M,\sigma):=h_\sigma^{1\over2q}L^p(\M)h_\sigma^{1\over2q}\ (\subseteq L^1(\M)), \\
&\|h_\sigma^{1\over2q}ah_\sigma^{1\over2q}\|_{p,\sigma}:=\|a\|_p,\qquad a\in L^p(\M),
\end{align*}
that is, $L^p(\M)\cong L^p(\M,\sigma)$ by the isometry
$a\mapsto h_\sigma^{1\over2q}ah_\sigma^{1\over2q}$
(see Appendix \ref{Sec-C}).

In this general setting, the sandwiched R\'enyi $\alpha$-divergence
$D_{\alpha}\nw(\rho\|\sigma)$
of $\rho,\sigma\in\M_*^+$ is defined by the same formula as in \eqref{F-2.1},
with $Q_{\alpha}\nw(\rho\|\sigma)$ in \eqref{F-2.2} replaced by
\begin{align}\label{F-2.3}
Q_\alpha^*(\rho\|\sigma):=\begin{cases}
\tr\bigl(h_\sigma^{1-\alpha\over2\alpha}h_\rho
h_\sigma^{1-\alpha\over2\alpha}\bigr)^\alpha
=\|h_\sigma^{1-\alpha\over2\alpha}h_\rho^{1/2}\|_{2\alpha}^{2\alpha},&
\text{if}\ \alpha\in[1/2,1),\\
\|h_\rho\|_{\alpha,\sigma}^\alpha, &
\text{if}\ \alpha>1\ \text{and}\ h_\rho\in L^\alpha(\M,\sigma), \\
+\infty, & \text{otherwise};
\end{cases}
\end{align}
according to \cite{Je1} for $\alpha>1$, and \cite[Theorem 3.1]{Je2} and \cite[Theorem 3.11]{Hi} for
$\alpha\in[1/2,1)$.
Note here that the condition $h_\rho\in L^\alpha(\M,\sigma)$ contains,
in particular, $s(\rho)\le s(\sigma)$.

\begin{example}\label{E-2.2}\rm
Consider the simple case where $\M=B(\cH)$ with a finite-dimensional Hilbert space $\cH$.
Then the larger von Neumann algebra $\N$ mentioned above is given as
$\N=B(\cH)\overline\otimes L^\infty(\bR)$ and Haagerup's $L^p$-spaces are given by
\[
L^p(\M)=\begin{cases}
B(\cH)\otimes e^{-t/p}, & p\in[1,+\infty), \\
B(\cH)\otimes\1=B(\cH), & p=+\infty,
\end{cases}
\]
with norms
\[
\begin{cases}
\|X\otimes e^{-t/p}\|_p=\|X\|_p=(\Tr|X|^p)^{1/p}, & p\in[1,+\infty), \\
\|X\otimes\1\|_\infty=\|X\|_\infty\ \,(\mbox{operator norm}), & p=+\infty,
\end{cases}
\]
where $e^{-t/p}$ is a shorthand notation for the multiplication operator on $L^2(\bR)$ with the
function $t\mapsto e^{-t/p}$ on $\bR$. Then $h_{\psi}=\hat\psi\otimes e^{-t}$ for any
$\psi\in\M_*^+$, and a straightforward computation yields that the definitions in
\eqref{F-2.2} and in \eqref{F-2.3} give the same values when $\alpha\in[1/2,1)$.
On the other hand, for any $\sigma\in\M_*^+$ with $e:=s(\sigma)$, Kosaki's interpolation
$L^p$-spaces are given by
\[
L^p(\M,\sigma)=\hat\sigma^{1\over2q}B(\cH)\hat\sigma^{1\over2q}=eB(\cH)e,
\qquad p\in[1,+\infty],\ 1/p+1/q=1,
\]
with norms
\[
\|X\|_{p,\sigma}=\|\hat\sigma^{-{1\over2q}}X\hat\sigma^{-{1\over2q}}\|_p,
\qquad X\in eB(\cH)e.
\]
Then it immediately follows that the definitions in \eqref{F-2.2} and in \eqref{F-2.3}
give the same values when $\alpha>1$ too. In this way, \eqref{F-2.3} does indeed give an
extension of the definition \eqref{F-2.1} from the finite-dimensional to the most
general case; see Appendices \ref{Sec-B} and \ref{Sec-C} for more about Haagerup's and
Kosaki's $L^p$-spaces in the case $\M=B(\cH)$.
\end{example}

When $\rho\in\M_*^+$ is a state, $\alpha\mapsto D_\alpha^*(\rho\|\sigma)$ is monotone
increasing on $[1/2,1)\cup(1,+\infty)$, and
\begin{align}\label{F-2.4}
\lim_{\alpha\nearrow1}D_\alpha^*(\rho\|\sigma)=D(\rho\|\sigma);
\end{align}
if, in addition, $D_\alpha^*(\rho\|\sigma)<+\infty$ for some
$\alpha>1$, then
\begin{align}\label{F-2.5}
\lim_{\alpha\searrow1}D_\alpha^*(\rho\|\sigma)=D(\rho\|\sigma),
\end{align}
where $D(\rho\|\sigma)$ is the \emph{relative entropy} of $\rho$ and $\sigma$
\cite{Ar1,Ar2,Ko3,Umegaki}. Moreover, for any $\rho,\sigma\in\M_*^+$,
\begin{align}\label{F-2.6}
\lim_{\alpha\to+\infty}D_\alpha^*(\rho\|\sigma)=D_{\max}(\rho\|\sigma),
\end{align}
where
\[
D_{\max}(\rho\|\sigma):=\log\min\{\lambda\ge0:\rho\le\lambda\sigma\}
\]
($=+\infty$ if no such $\lambda$ exists), is the \emph{max-relative entropy}
\cite{Da,RennerPhD}. For these properties of $D_\alpha^*$, see \cite{MDSFT,WWY} for the
finite-dimensional case and \cite{BST,Je1,Je2} (also a concise survey in \cite[Sec.~3.3]{Hi}) for
the von Neumann algebra case.

The next variational formulas shown in \cite[Lemma 3.19]{Hi} and \cite[Proposition 3.4]{Je2}
are the von Neumann algebra versions of \cite[Lemma 4]{FL}, which will play a crucial role in the
next section. Here, $\M_+$ is the set of positive operators in $\M$ and $\M_{++}$ is the set
of invertible $x\in\M_+$.

\begin{proposition}[\cite{Hi,Je2}]\label{P-2.3}
For any $\rho,\sigma\in\M_*^+$ the following hold:
\begin{itemize}
\item[\rm(i)] For every $\alpha\in(1,+\infty)$,
\begin{align}\label{F-2.7}
Q_\alpha^*(\rho\|\sigma)=\sup_{x\in\M_+}\Bigl[\alpha\rho(x)
-(\alpha-1)\tr\bigl(h_\sigma^{\alpha-1\over2\alpha}xh_\sigma^{\alpha-1\over2\alpha}
\bigr)^{\alpha\over\alpha-1}\Bigr].
\end{align}
\item[\rm(ii)] For every $\alpha\in[1/2,1)$,
\begin{align}\label{F-2.8}
Q_\alpha^*(\rho\|\sigma)=\inf_{x\in\M_{++}}\Bigl[\alpha\rho(x)
+(1-\alpha)\tr\bigl(h_\sigma^{1-\alpha\over2\alpha}x^{-1}h_\sigma^{1-\alpha\over2\alpha}
\bigr)^{\alpha\over1-\alpha}\Bigr].
\end{align}
\end{itemize}
\end{proposition}
\medskip

A different quantum extension of the classical R\'enyi divergences is given by the
\emph{standard} (or \emph{Petz type}) \emph{R\'enyi divergences} $D_\alpha(\rho\|\sigma)$
defined for every $\rho,\sigma\in\M_*^+$ and any $\alpha\in[0,+\infty)\setminus\{1\}$ in terms of
the \emph{relative modular operator} $\Delta_{\rho,\sigma}$ (see Appendix \ref{Sec-A}),
which is a particular case of \emph{standard $f$-divergences} developed first in \cite{Ko2,Pe1}.
The following brief overview is based on \cite{Hi}. When $0\le\alpha<1$, note that
$h_\sigma^{1/2}$ is in the domain $\cD(\Delta_{\rho,\sigma}^{\alpha/2})$ of
$\Delta_{\rho,\sigma}^{\alpha/2}$, and define
\begin{align}\label{F-2.9}
Q_\alpha(\rho\|\sigma):=\|\Delta_{\rho,\sigma}^{\alpha/2}h_\sigma^{1/2}\|^2.
\end{align}
When $\alpha>1$,
\begin{align}\label{F-2.10}
Q_\alpha(\rho\|\sigma):=\begin{cases}
\|\Delta_{\rho,\sigma}^{\alpha/2}h_\sigma^{1/2}\|^2 & \text{if $s(\rho)\le s(\sigma)$ and $h_\sigma^{1/2}\in\cD(\Delta_{\rho,\sigma}^{\alpha/2})$}, \\
+\infty & \text{otherwise}.\end{cases}
\end{align}
Then $D_\alpha(\rho\|\sigma)$ is defined as
\[
D_\alpha(\rho\|\sigma):={1\over\alpha-1}\log Q_\alpha(\rho\|\sigma).
\]

\begin{example}\label{E-2.4}\rm
Assume that $\M=B(\cH)$ with $\dim\cH<+\infty$. For any $\rho,\sigma\in B(\cH)^+$,
since $\Delta_{\rho,\sigma}=L_\rho R_{\sigma^{-1}}$ (see Appendix \ref{Sec-A}) and hence
$\Delta_{\rho,\sigma}^{\alpha/2}\sigma^{1/2}=\rho^{\alpha/2}\sigma^{1-\alpha\over2}$, where
$\sigma^{1-\alpha\over2}$ is defined with restriction to the support $s(\sigma)\cH$ when
$\alpha>1$. Thus, the expressions of $Q_\alpha(\rho\|\sigma)$ in \eqref{F-2.9} and
\eqref{F-2.10} give the well-known formulas \cite{Petz}
\[
Q_\alpha(\rho\|\sigma)=\begin{cases}
\Tr(\rho^\alpha\sigma^{1-\alpha}), & \text{if $0\le\alpha<1$ or $s(\rho)\le s(\sigma)$}, \\
+\infty, & \text{if $\alpha>1$ and $s(\rho)\not\le s(\sigma)$}.
\end{cases}
\]
\end{example}

Properties of $D_\alpha(\rho\|\sigma)$ in the von Neumann algebra case were summarized in
\cite[Proposition 5.3]{Hi3}, and a handy description of $D_\alpha(\rho\|\sigma)$ in terms of
$h_\rho,h_\sigma$ was given in \cite[Theorem 3.6]{Hi}; in particular, when $\alpha\in[0,1)$,
\[
D_\alpha(\rho\|\sigma)={1\over\alpha-1}\log\tr(h_\rho^\alpha h_\sigma^{1-\alpha}).
\]
(Compare this with \eqref{F-2.3}.) When $\rho$ is a state, 
$\alpha\mapsto D_\alpha(\rho\|\sigma)$ is monotone increasing on $[0,1)\cup(1,+\infty)$,
$\lim_{\alpha\nearrow1}D_\alpha(\rho\|\sigma)=D(\rho\|\sigma)$, and if
$D_\alpha(\rho\|\sigma)<+\infty$ for some $\alpha>1$, then
$\lim_{\alpha\searrow1}D_\alpha(\rho\|\sigma)=D(\rho\|\sigma)$.
The $\alpha=0$ case is the \emph{min-relative entropy} \cite{Da,RennerPhD}
\[
D_{\min}(\rho\|\sigma):=D_0(\rho\|\sigma)=-\log\tr(s(\rho)h_\sigma).
\]
According to \cite[Theorem 1.2]{BST} and 
\cite[Corollary 3.6]{Je1}, the inequality
\[
D_\alpha^*(\rho\|\sigma)\le D_\alpha(\rho\|\sigma)
\]
holds for every $\rho,\sigma\in\M_*^+$ and any $\alpha\in[1/2,+\infty)\setminus\{1\}$, while
equality holds here when $\rho,\sigma$ ``commute'' (see \cite[Remark 3.18\,(2)]{Hi} for the
precise statement). Apart from $D_\alpha$ and $D_\alpha^*$, the two extreme cases
$D_{\min}$ and $D_{\max}$ are also useful in some quantum information problems such as
resource theory; see, e.g., \cite{BuSuTo,SFKMNB,Wa-Wi}.

\section{The von Neumann algebra case}\label{Sec-3}

\subsection{Martingale convergence for sandwiched R\'enyi divergences}\label{Sec-3.1}

Let $\M$ be a von Neumann algebra and $\{\M_i\}_{i\in\cI}$ be an increasing net (on
a directed set $\cI$) of von Neumann subalgebras of $\M$ containing the unit of $\M$,
such that $\M$ is generated by $\bigcup_{i\in\cI}\M_i$, i.e.,
\[
\M=\biggl(\bigcup_{i\in\cI}\M_i\biggr)''.
\]
Let $\rho,\sigma\in\M_*^+$ and $\rho_i:=\rho|_{\M_i}$,
$\sigma_i:=\sigma|_{\M_i}$ for each $i\in\cI$.

The next theorem provides the \emph{martingale convergence} for the sandwiched R\'enyi
divergence $D_\alpha^*$. It will play an essential role repeatedly in our later discussions.

\begin{theorem}\label{T-3.1}
Let $\{\M_i\}_{i\in\cI}$ be as stated above, and assume that $\M$ is $\sigma$-finite. Then for
every $\rho,\sigma\in\M_*^+$ and for every $\alpha\in[1/2,+\infty)\setminus\{1\}$ we have
\begin{align}\label{F-3.1}
D_\alpha^*(\rho\|\sigma)=\lim_i D_\alpha^*(\rho_i\|\sigma_i)\quad\mbox{increasingly}.
\end{align}
\end{theorem}

To prove the theorem, we first give two lemmas. First, we state the martingale convergence for
the generalized conditional expectations in \cite[Theorem 3]{HT1} as a lemma, which was given in
\cite{HT1} in a slightly more general setting.

\begin{lemma}[\cite{HT1}]\label{L-3.2}
In the situation stated above, assume that $\sigma$ is faithful and for each $i\in\cI$ let
$\cE_{\M_i,\sigma}:\M\to\M_i$ be the generalized conditional expectation with respect to
$\sigma$ (see Section \ref{Sec-D}). Then for every $x\in M$ we have $\cE_{\M_i,\sigma}(x)\to x$
strongly.
\end{lemma}

The next lemma is indeed a special case of \cite[Proposition 3.12]{Je1}. The argument is also
found in \cite[Sec.~2.7]{JL}. We supply a sketchy proof for the convenience of the reader.

\begin{lemma}[\cite{Je1}]\label{L-3.3}
Assume that $\sigma\in\M_*^+$ is faithful. Let $\N$ be a von Neumann subalgebra of $\M$
containing the unit of $\M$, and $\sigma_0:=\sigma|_{\N}$. Let $\cE_{\N,\sigma}:\M\to\N$
be the generalized conditional expectation with respect to $\sigma$. Then for every $p\in[1,+\infty)$
and $x\in\M_+$ we have
\[
\tr\bigl(h_{\sigma_0}^{1\over2p}\cE_{\N,\sigma}(x)h_{\sigma_0}^{1\over2p}\bigr)^p
\le\tr\bigl(h_\sigma^{1\over2p}xh_\sigma^{1\over2p}\bigr)^p,
\]
where $h_{\sigma_0}$ is the element of $L^1(\N)_+$ corresponding to $\sigma_0$, and
$h_\sigma\in L^1(\M)_+$ corresponds to $\sigma$.
\end{lemma}

\begin{proof}[Proof (sketch)]
We utilize Kosaki's (symmetric) interpolation $L^p$-space $L^p(\M,\sigma)$ with the norm
$\norm{\cdot}_{p,\sigma}$ for $p\in[1,+\infty]$; see Appendix \ref{Sec-C}. Let
$\Psi=\Phi_*:L^1(\M)\to L^1(\N)$ be the predual map of the injection $\N\hookrightarrow\M$,
i.e., $\Psi(h_\omega)=h_{\omega|_\N}$ for $\omega\in\M_*$, so that $\Psi$ is contractive
with respect to $\norm{\cdot}_1$. In the present situation, the description of $\Phi_\omega^*$ in
\eqref{F-D.3} of Appendix \ref{Sec-D} shows that $\Psi$ restricted to $L^\infty(\M,\sigma)$
is given as
\[
\Psi(h_\sigma^{1/2}xh_\sigma^{1/2})
=h_{\sigma_0}^{1/2}\cE_{\N,\sigma}(x)h_{\sigma_0}^{1/2},\qquad x\in\M,
\]
(so that the map $\cE_{\N,\sigma}$ coincides with $\Phi_\sigma:\M\to\N$ given in \cite{Je1}).
Hence $\Psi$ is contractive from $L^\infty(\M,\sigma)$ to $L^\infty(\N,\sigma_0)$ with respect to
$\norm{\cdot}_{\infty,\sigma}$ (this can be seen more directly by \eqref{F-D.4}). It follows from
the complex interpolation method (the Riesz--Thorin theorem) that $\Psi$ is a contraction from
$L^p(\M,\sigma)$ to $L^p(\N,\sigma_0)$ with respect to $\norm{\cdot}_{p,\sigma}$ for any
$p\in(1,+\infty)$. For
$h_\sigma^{1/2}xh_\sigma^{1/2}\in L^\infty(\M,\sigma)\subseteq L^p(\M,\sigma)$ we have
\[
\|h_{\sigma_0}^{1/2}\cE_{\N,\sigma}(x)h_{\sigma_0}^{1/2}\|_{p,\sigma}
\le\|h_\sigma^{1/2}xh_\sigma^{1/2}\|_{p,\sigma}.
\]
Noting by \eqref{F-C.1} and \eqref{F-C.2} that
\[
\|h_\sigma^{1/2}xh_\sigma^{1/2}\|_{p,\sigma}
=\|h_\sigma^{1\over2p}xh_\sigma^{1\over2p}\|_p
=\Bigl[\tr\bigl(h_\sigma^{1\over 2p}xh_\sigma^{1\over2p}\bigr)^p\Bigr]^{1/p}
\]
and similarly for $\|h_{\sigma_0}^{1/2}\cE_{\N,\sigma}(x)h_{\sigma_0}^{1/2}\|_{p,\sigma}$,
we have
\[
\tr\bigl(h_{\sigma_0}^{1\over2p}\cE_{\N,\sigma}(x)h_{\sigma_0}^{1\over2p}\bigr)^p
\le\tr\bigl(h_\sigma^{1\over2p}xh_\sigma^{1\over2p}\bigr)^p,
\]
as desired.
\end{proof}

\begin{proof}[Proof of Theorem \ref{T-3.1}]
From the monotonicity property of $D_\alpha^*$ proved in \cite{BST,Je1} it follows that
$D_\alpha^*(\rho_i\|\sigma_i)\le D_\alpha^*(\rho\|\sigma)$ and
$i\in\cI\mapsto D_\alpha^*(\rho_i\|\sigma_i)$ is increasing. Hence, to show \eqref{F-3.1}, it
suffices to prove that
\begin{align}\label{F-3.2}
D_\alpha^*(\rho\|\sigma)\le\sup_{i\in\cI}D_\alpha^*(\rho_i\|\sigma_i).
\end{align}
To do this, we may assume that $\sigma$ is faithful. Indeed, assume that \eqref{F-3.2} has
been shown when $\sigma$ is faithful. For general $\sigma\in\M_*^+$, since $\M$ is
$\sigma$-finite, there exists a $\sigma_0\in\M_*^+$ with $s(\sigma_0)=1-s(\sigma)$ and let
$\sigma^{(n)}:=\sigma+n^{-1}\sigma_0$, $\sigma_i^{(n)}:=\sigma^{(n)}|_{\M_i}$. From the lower
semi-continuity and the order relation of $D_\alpha^*$ (see \cite{Je1}, \cite[Theorem 3.16]{Hi})
it follows that
\[
D_\alpha^*(\rho\|\sigma)\le\liminf_{n\to\infty}D_\alpha^*(\rho\|\sigma^{(n)})
\le\liminf_{n\to\infty}\sup_iD_\alpha^*(\rho_i\|\sigma_i^{(n)})
\le\sup_iD_\alpha^*(\rho_i\|\sigma_i),
\]
proving \eqref{F-3.2} for general $\sigma$. Below we assume the faithfulness of $\sigma$ and
divide the proof into two cases $1<\alpha<+\infty$ and $1/2\le\alpha<1$.

{\it Case $1<\alpha<+\infty$}.\enspace
We need to prove that
\begin{align}\label{F-3.3}
Q_\alpha^*(\rho\|\sigma)\le\sup_{i\in\cI}Q_\alpha^*(\rho_i\|\sigma_i).
\end{align}
For every $x\in\M_+$ and $i\in\cI$,  by Lemma \ref{L-3.3} we have
\begin{align*}
&\alpha\rho(x)-(\alpha-1)\tr\bigl(h_\sigma^{\alpha-1\over2\alpha}
xh_\sigma^{\alpha-1\over2\alpha}\bigr)^{\alpha\over\alpha-1} \\
&\quad\le\alpha\rho(x)-(\alpha-1)\tr\bigl(h_{\sigma_i}^{\alpha-1\over2\alpha}
\cE_{\M_i,\sigma}(x)h_{\sigma_i}^{\alpha-1\over2\alpha}\bigr)^{\alpha\over\alpha-1} \\
&\quad=\alpha\rho(x-\cE_{\M_i,\sigma}(x))
+\alpha\rho_i(\cE_{\M_i,\sigma}(x))-(\alpha-1)\tr\bigl(h_{\sigma_i}^{\alpha-1\over2\alpha}
\cE_{\M_i,\sigma}(x)h_{\sigma_i}^{\alpha-1\over2\alpha}\bigr)^{\alpha\over\alpha-1} \\
&\quad\le\alpha\rho(x-\cE_{\M_i,\sigma}(x))+\sup_{j\in\cI}Q_\alpha^*(\rho_j\|\sigma_j),
\end{align*}
where the last inequality is due to the variational formula \eqref{F-2.7}. Lemma \ref{L-3.2} gives
\[
\alpha\rho(x)-(\alpha-1)\tr\bigl(h_\sigma^{\alpha-1\over2\alpha}
xh_\sigma^{\alpha-1\over2\alpha}\bigr)^{\alpha\over\alpha-1}
\le\sup_{j\in\cI}Q_\alpha^*(\rho_j\|\sigma_j),\qquad x\in M_+.
\]
By \eqref{F-2.7} again we have \eqref{F-3.3}.

{\it Case $1/2\le\alpha<1$}.\enspace
We need to prove that
\begin{align}\label{F-3.4}
Q_\alpha^*(\rho\|\sigma)\ge\inf_{i\in\cI}Q_\alpha^*(\rho_i\|\sigma_i).
\end{align}
For every $x\in\M_{++}$ and $i\in\cI$, by Lemma \ref{L-3.3} we have
\begin{align*}
&\alpha\rho(x)+(1-\alpha)\tr\bigl(h_\sigma^{1-\alpha\over2\alpha}
x^{-1}h_\sigma^{1-\alpha\over2\alpha}\bigr)^{\alpha\over1-\alpha} \\
&\quad\ge\alpha\rho(x)+(1-\alpha)\tr\bigl(h_{\sigma_i}^{1-\alpha\over2\alpha}
\cE_{\M_i,\sigma}(x^{-1})h_{\sigma_i}^{1-\alpha\over2\alpha}\bigr)^{\alpha\over1-\alpha} \\
&\quad\ge\alpha\rho(x)+(1-\alpha)\tr\bigl(h_{\sigma_i}^{1-\alpha\over2\alpha}
\cE_{\M_i,\sigma}(x)^{-1}h_{\sigma_i}^{1-\alpha\over2\alpha}\bigr)^{\alpha\over1-\alpha} \\
&\quad=\alpha\rho(x-\cE_{\M_i,\sigma}(x))
+\alpha\rho_i(\cE_{\M_i,\sigma}(x))+(1-\alpha)\tr\bigl(h_{\sigma_i}^{1-\alpha\over2\alpha}
\cE_{\M_i,\sigma}(x)^{-1}h_{\sigma_i}^{1-\alpha\over2\alpha}\bigr)^{\alpha\over1-\alpha} \\
&\quad\ge\alpha\rho(x-\cE_{\M_i,\sigma}(x))+\inf_{j\in\cI}Q_\alpha^*(\rho_j\|\sigma_j),
\end{align*}
where the second inequality above follows from the Jensen inequality
$\cE_{\M_i,\sigma}(x^{-1})\ge\cE_{\M_i,\sigma}(x)^{-1}$ (see \cite[Corollary 2.3]{Ch}), and
the last inequality is due to \eqref{F-2.8}. By Lemma \ref{L-3.2} and \eqref{F-2.8} we have
\eqref{F-3.4}.
\end{proof}

From Theorem \ref{T-3.1} we can easily obtain the following martingale type convergence for
$D_\alpha^*$ under the restriction to reduced subalgebras $e_i\M e_i$ with $e_i\nearrow1$. 
See \cite{Mo} for related results in the case $\M=\B(\hil)$.

\begin{proposition}\label{P-3.4}
Assume that $\M$ is $\sigma$-finite. Let $\{e_i\}_{i\in\cI}$ be an increasing net of projections
in $\M$ such that $e_i\nearrow 1$. Then for every $\rho,\sigma\in\M_*^+$ and every
$\alpha\in[1/2,+\infty)\setminus\{1\}$ we have
\[
D_\alpha^*(\rho\|\sigma)=\lim_iD_\alpha^*(e_i\rho e_i\|e_i\sigma e_i)\quad
\mbox{increasingly},
\]
where $e_i\rho e_i$ is the restriction of $\rho$ to the reduced von Neumann algebra $e_i\M e_i$
and similarly for $e_i\sigma e_i$.
\end{proposition}

\begin{proof}
Let $\M_i:=e_i\M e_i\oplus\bC(1-e_i)$; then $\{\M_i\}_{i\in\cI}$ is an increasing net of  von
Neumann subalgebras of $\M$ containing the unit of $\M$ with $\M=\bigl(\bigcup_i\M_i\bigr)''$.
Note that $\rho_i:=\rho|_{\M_i}=e_i\rho e_i\oplus\rho(1-e_i)$ and similarly for
$\sigma_i:=\sigma|_{\M_i}$. Hence by the definition of $Q_\alpha^*$ or by using the variational
formulas in Proposition \ref{P-2.3}, for any $\alpha\in[1/2,+\infty)\setminus\{1\}$ we have
\[
Q_\alpha^*(\rho_i\|\sigma_i)
=Q_\alpha^*(e_i\rho e_i\|e_i\sigma e_i)+\rho(1-e_i)^\alpha\sigma(1-e_i)^{1-\alpha}.
\]
Below we will give the proof for the case $\alpha>1$ (that for $1/2\le\alpha<1$ is similar by using
the reverse monotonicity of $Q_\alpha^*$). Let $i,j\in\cI$ be such that $i\le j$. Noting that
$e_i\M e_i\oplus\bC(e_j-e_i)\subseteq e_j\M e_j$, by the monotonicity of $Q_\alpha^*$ we have
\begin{align*}
Q_\alpha^*(e_j\rho e_j\|e_j\sigma e_j)
&\ge Q_\alpha^*(e_i\rho e_i\|e_i\sigma e_i)+\rho(e_j-e_i)^\alpha\sigma(e_j-e_i)^{1-\alpha} \\
&\ge Q_\alpha^*(e_i\rho e_i\|e_i\sigma e_i).
\end{align*}
Hence, $i\in\cI\mapsto Q_\alpha^*(e_i\rho e_i\|e_i\sigma e_i)$ is increasing. Since Theorem
\ref{T-3.1} implies that
\[
Q_\alpha^*(\rho\|\sigma)=\lim_i\bigl[Q_\alpha^*(e_i\rho e_i\|e_i\sigma e_i)
+\rho(1-e_i)^\alpha\sigma(1-e_i)^{1-\alpha}\bigr],
\]
the assertion follows if we show the following:
\begin{itemize}
\item when $Q_\alpha^*(\rho\|\sigma)=+\infty$, $\lim_iQ_\alpha^*(e_i\rho e_i\|e_i\sigma e_i)=+\infty$,
\item when $Q_\alpha^*(\rho\|\sigma)<+\infty$, $\lim_i\rho(1-e_i)^\alpha\sigma(1-e_i)^{1-\alpha}=0$.
\end{itemize}
These can indeed be proved in a similar way to the proof of \cite[Theorem 4.5]{Hi3} by taking
$\rho(1-e_i)^\alpha\sigma(1-e_i)^{1-\alpha}$ and $\rho(e_j-e_i)^\alpha\sigma(e_j-e_i)^{1-\alpha}$ in
place of $\rho(1-e_i)f(\rho(1-e_i)/\sigma(1-e_i))$ and $\rho(e_j-e_i)f(\rho(e_j-e_i)/\sigma(e_j-e_i))$,
respectively, there. The details are omitted here.
\end{proof}

\subsection{The strong converse exponent in injective von Neumann algebras}\label{Sec-3.2}

Let $\M$ be a von Neumann algebra and $\rho,\sigma\in\M_*^+$ be non-zero. For each
$n\in\bN$ let $\M^{\overline\otimes n}$ be the $n$-fold von Neumann algebra tensor product of
$\M$, and $\rho_n:=\rho^{\otimes n}$ (resp., $\sigma_n:=\sigma^{\otimes n}$) be the $n$-fold
tensor product of $\rho$ (resp., $\sigma$), which are elements of
$\bigl(\M^{\overline\otimes n}\bigr)_*^+$; see \cite[Chap.~IV]{Ta1}.

In this section we consider the simple hypothesis testing problem for the \emph{null hypothesis}
$H_0:\rho$ versus the \emph{alternative hypothesis} $H_1:\sigma$, and extend the result
\cite{MO1} on the strong converse exponent to the injective von Neumann algebra case.
For a test $T_n\in\M^{\overline\otimes n}$ with $0\le T_n\le1$, $\rho_n(1-T_n)$ and
$\sigma_n(T_n)$ represent the \emph{type I} and the \emph{type II error probabilities},
respectively (hence $\rho_n(T_n)$ is the \emph{type I success probability}) when $\rho,\sigma$
are states. We begin by defining several forms of the strong converse exponents in the following:

\begin{definition}\label{D-3.5}\rm
For each type II error exponent $r\in\bR$ we define the \emph{strong converse exponents} of
simple hypothesis testing for $H_0:\rho$ vs.\ $H_1:\sigma$ as follows:
\begin{align*}
\underline{sc}_r(\rho\|\sigma)&:=\inf_{\{T_n\}}\Bigl\{\liminf_{n\to\infty}-{1\over n}\log\rho_n(T_n):
\liminf_{n\to\infty}-{1\over n}\log\sigma_n(T_n)\ge r\Bigr\}, \\
\overline{sc}_r(\rho\|\sigma)&:=\inf_{\{T_n\}}\Bigl\{\limsup_{n\to\infty}-{1\over n}\log\rho_n(T_n):
\liminf_{n\to\infty}-{1\over n}\log\sigma_n(T_n)\ge r\Bigr\}, \\
sc_r(\rho\|\sigma)&:=\inf_{\{T_n\}}\Bigl\{\lim_{n\to\infty}-{1\over n}\log\rho_n(T_n):
\liminf_{n\to\infty}-{1\over n}\log\sigma_n(T_n)\ge r\Bigr\},
\end{align*}
where the infima  are taken over all test sequences $\{T_n\}$ with $T_n\in\M^{\overline\otimes n}$,
$0\le T_n\le1$ ($n\in\bN$) for which the indicated condition holds (and furthermore the limit exists
for $sc_r(\rho\|\sigma)$). Also, we write $\underline{sc}_r^0(\rho\|\sigma)$,
$\overline{sc}_r^0(\rho\|\sigma)$ and $sc_r^0(\rho\|\sigma)$ for the above infima when the
condition $\liminf_{n\to\infty}-{1\over n}\log\sigma_n(T_n)\ge r$ is replaced with
$\liminf_{n\to\infty}-{1\over n}\log\sigma_n(T_n)>r$.
\end{definition}

For any $r\in\bR$ it is obvious that
\begin{align}\label{F-3.5}
\begin{array}{lllll}
\underline{sc}_r(\rho\|\sigma) &\le & \overline{sc}_r(\rho\|\sigma) &\le & sc_r(\rho\|\sigma), \\
\ds\ds\vertleq & & \ds\ds\vertleq & & \ds\ds\vertleq \\
\underline{sc}_r^0(\rho\|\sigma) & \le & \overline{sc}_r^0(\rho\|\sigma) & \le& sc_r^0(\rho\|\sigma).
\end{array}
\end{align}

The following is the definition of the Hoeffding anti-divergence \cite{MO1,Mo} in the von Neumann
algebra setting.

\begin{definition}\label{D-3.6}\rm
For any $\rho,\sigma\in\M_*^+$ define
\begin{align*}
\psi^*(\rho\|\sigma|\alpha)
&:=\log Q_\alpha^*(\rho\|\sigma)=(\alpha-1)D_\alpha^*(\rho\|\sigma),
\qquad\alpha\in(1,+\infty), \\
\tilde\psi^*(\rho\|\sigma|u)
&:=(1-u)\psi^*(\rho\|\sigma|(1-u)^{-1})=uD_{1\over1-u}^*(\rho\|\sigma),
\qquad u\in(0,1).
\end{align*}
The \emph{Hoeffding anti-divergence} of $\rho$ and $\sigma$ is then defined for each $r\in\bR$ by
\begin{align}\label{F-3.6}
H_r^*(\rho\|\sigma)&:=\sup_{\alpha>1}{\alpha-1\over\alpha}\bigl\{r-D_\alpha^*(\rho\|\sigma)\bigr\}
=\sup_{u\in(0,1)}\bigl\{ur-\tilde\psi^*(\rho\|\sigma|u)\bigr\}.
\end{align}
\end{definition}

The aim of this section is to find whether the strong converse exponents in \eqref{F-3.5} are all equal
to $H_r^*(\rho\|\sigma)$ for given $r$, as in the finite-dimensional case. To do so, we may assume without loss of generality that
$\rho,\sigma\in\M_*^+$ are states. In fact, for any $\lambda,\mu>0$ it is immediate to check that
if $\eps_r(\rho\|\sigma)$ is any strong converse exponent in \eqref{F-3.5}, then
\[
\eps_r(\lambda\rho\|\mu\sigma)=\eps_{r+\log\mu}(\rho\|\sigma)-\log\lambda.
\]
On the other hand, since
$Q_\alpha^*(\lambda\rho\|\mu\sigma)=\lambda^\alpha\mu^{1-\alpha}Q_\alpha^*(\rho\|\sigma)$,
one easily sees that
\begin{align}\label{F-3.7}
H_r^*(\lambda\rho\|\mu\sigma)=H_{r+\log\mu}^*(\rho\|\sigma)-\log\lambda.
\end{align}
Thus, in the rest of the section, we will always assume that $\rho,\sigma\in\M_*^+$ are states.
Under this assumption, when $r<0$, any exponent in \eqref{F-3.5} is clearly equal to $0$ by
taking $T_n=1$ for all $n$, while $H_r^*(\rho\|\sigma)$ is also $0$ whenever
$D_\alpha^*(\rho\|\sigma)<+\infty$ for some $\alpha>1$, otherwise $H_r^*(\rho\|\sigma)=-\infty$
for all $r\in\bR$. Therefore, when $r<0$, the conjecture holds in a trivial way, or otherwise it is
not true. Furthermore, when $s(\rho)\not\le s(\sigma)$, any exponent in \eqref{F-3.5} is $0$ for
all $r\in\bR$ by taking $T_n=1-s(\sigma_n)=1-s(\sigma)^{\otimes n}$, while
$H_r^*(\rho\|\sigma)=-\infty$ for all $r$ by definition. So the conjecture is not true for any $r$ in
this case. Therefore, we may restrict our consideration to the case where $r\ge0$ and
$s(\rho)\le s(\sigma)$.

When $\M$ is finite-dimensional and $s(\rho)\le s(\sigma)$, it was proved in
\cite[Theorem 4.10]{MO1} that $\underline{sc}_r(\rho\|\sigma)=H_r^*(\rho\|\sigma)$ for any
$r\ge0$. (Note that $\underline{sc}_r(\rho\|\sigma)$ here was denoted by $B_e^*(r)$ in \cite{MO1}.)
The result has recently been extended in \cite[Theorem IV.5]{Mo} to the infinite-dimensional
$\B(\cH)$ setting in such a way that
$\underline{sc}_r(\rho\|\sigma)=\overline{sc}_r(\rho\|\sigma)=H_r^*(\rho\|\sigma)$ for any
$r\in\bR$ under the assumption that $D_\alpha^*(\rho\|\sigma)<+\infty$ for some $\alpha>1$.
The main aim of this section is to further extend the result to the injective von Neumann algebra
setting. For the  convenience of the reader we recall the fundamental notions of injectivity
and AFD for von Neumann algebras in Appendix \ref{Sec-E}.

\begin{theorem}\label{T-3.7}
Assume that $\M$ is injective. Let $\rho,\sigma\in\M_*^+$ be states such that
$D_\alpha^*(\rho\|\sigma)<+\infty$ for some $\alpha>1$ (in particular, $s(\rho)\le s(\sigma)$). Then
for every $r\in\bR$ we have
\begin{align}\label{F-3.8}
\underline{sc}_r(\rho\|\sigma)=\overline{sc}_r(\rho\|\sigma)=sc_r(\rho\|\sigma)
=\underline{sc}_r^0(\rho\|\sigma)=\overline{sc}_r^0(\rho\|\sigma)=sc_r^0(\rho\|\sigma)
=H_r^*(\rho\|\sigma).
\end{align}
\end{theorem}

\begin{proof}
To prove this, we may and do assume that $\M$ is $\sigma$-finite. Indeed, let
$e:=s(\rho)\vee s(\sigma)$. Then $e\M e$ is injective and all the quantities in \eqref{F-3.8} for
any $r\in\bR$ as well as $D_\alpha^*(\rho\|\sigma)$ for any $\alpha\in[1/2,+\infty)\setminus\{1\}$
are left unchanged when $\rho,\sigma$ are replaced with $\rho|_{e\M e},\sigma|_{e\M e}$.

When $r<0$, the quantities in \eqref{F-3.8} are all equal to $0$ as explained in the paragraph
after Definition \ref{D-3.6}. Hence we may assume that $r\ge0$. By \eqref{F-3.5} it suffices to
show the following two inequalities
\begin{align}
\underline{sc}_r(\rho\|\sigma)&\ge H_r^*(\rho\|\sigma), \label{F-3.9}\\
sc_r^0(\rho\|\sigma)&\le H_r^*(\rho\|\sigma). \label{F-3.10}
\end{align}
The inequality \eqref{F-3.9} follows in the same way (using Nagaoka's method \cite{Na_sc}) as
the proof of \cite[Lemma 4.7]{MO1} in the finite-dimensional case, as mentioned in
\cite[Sec.~5]{BST}. The injectivity assumption is unnecessary for this part. As for the inequality
\eqref{F-3.10}, we need to first prove it in the finite-dimensional case, because the exponent
$sc_r^0(\rho\|\sigma)$ was not treated in \cite{MO1}. To make the proof of this part more
understandable, we present it in Appendix \ref{Sec-F} separately.

Now, by the injectivity assumption of $\M$ (see Appendix \ref{Sec-E}), we can choose an
increasing net $\{\M_i\}_{i\in\cI}$ of finite-dimensional *-subalgebras containing the unit of $\M$
such that $\M=\bigl(\bigcup_{i\in\cI}\M_i\bigr)''$. For each $i\in\cI$ let $\rho_i:=\rho|_{\M_i}$
and $\sigma_i:=\sigma|_{\M_i}$. According to \cite[Corollary 3.11]{MO1},
$\alpha\mapsto\psi^*(\rho_i\|\sigma_i|\alpha)$ is  a finite-valued convex function on $(1,+\infty)$,
whence it is also continuous there. Thus, the function
\[
u\,\longmapsto\,\tilde\psi^*(\rho_i\|\sigma_i|u)=(1-u)\psi^*(\rho_i\|\sigma_i|(1-u)^{-1})
\]
is also finite-valued, convex and continuous on $(0,1)$. In particular, it can be extended to a
convex and continuous function  on $[0,1]$ by
\[
\tilde\psi^*(\rho_i\|\sigma_i|0):=\lim_{u\searrow0}\tilde\psi^*(\rho_i\|\sigma_i|u)
=0,\qquad
\tilde\psi^*(\rho_i\|\sigma_i|1):=\lim_{u\nearrow1}\tilde\psi^*(\rho_i\|\sigma_i|u)
=D_{\max}(\rho_i\|\sigma_i),
\]
where the first equality is simple, and see \eqref{F-2.6} for the second one. We hence have
\[
H_r^*(\rho_i\|\sigma_i)=\sup_{u\in(0,1)}\bigl\{ur-\tilde\psi^*(\rho_i\|\sigma_i|u)\bigr\}
=\max_{u\in[0,1]}\bigl\{ur-\tilde\psi^*(\rho_i\|\sigma_i|u)\bigr\}.
\]
Also, from the monotonicity of $D_\alpha^*$, note that
$i\in\cI\mapsto\psi^*(\rho_i\|\sigma_i|\alpha)$ is increasing for any $\alpha>1$, so that
$i\in\cI\mapsto ur-\tilde\psi^*(\rho_i\|\sigma_i|u)$ is decreasing for any $u\in[0,1]$. For every
$r\ge0$, since $sc_r^0(\rho\|\sigma)\le sc_r^0(\rho_i\|\sigma_i)$ for all $i\in\cI$ as immediately
verified, we have
\begin{align}
sc_r^0(\rho\|\sigma)&\le\inf_{i\in\cI}sc_r^0(\rho_i\|\sigma_i)
=\inf_{i\in\cI}H_r^*(\rho_i\|\sigma_i) \nonumber\\
&=\inf_{i\in\cI}\max_{u\in[0,1]}\bigl\{ur-\tilde\psi^*(\rho_i\|\sigma_i|u)\bigr\} \nonumber\\
&=\max_{u\in[0,1]}\inf_{i\in\cI}\bigl\{ur-\tilde\psi^*(\rho_i\|\sigma_i|u)\bigr\} \nonumber\\
&=\max_{u\in[0,1]}\Bigl\{ur-\sup_i\tilde\psi^*(\rho_i\|\sigma_i|u)\Bigr\}, \label{F-3.11}
\end{align}
where the above first equality is due to Proposition \ref{P-F.2} in Appendix \ref{Sec-F} and the
third equality is due to a minimax theorem in \cite[Lemma II.3]{Mo} (also see \cite[Corollary A.2]{MH}).

Furthermore, Theorem \ref{T-3.1} implies that
\begin{align}\label{F-3.12}
\sup_{i\in\cI}\tilde\psi^*(\rho_i\|\sigma_i|u)
=\sup_{i\in\cI}{\alpha-1\over\alpha}\,D_\alpha^*(\rho_i\|\sigma_i)
={\alpha-1\over\alpha}\,D_\alpha^*(\rho\|\sigma)
=\tilde\psi^*(\rho\|\sigma|u)
\end{align}
for every $u\in(0,1)$ and $\alpha=(1-u)^{-1}\in(1,+\infty)$. By assumption,
$\tilde\psi^*(\rho\|\sigma|u)<+\infty$ for some $u\in(0,1)$. Hence by Lemma \ref{L-G.1} in
Appendix \ref{Sec-G} we confirm that
\begin{align}\label{F-3.13}
\sup_{i\in\cI}\tilde\psi^*(\rho_i\|\sigma_i|0)=\lim_{u\searrow0}\tilde\psi^*(\rho\|\sigma|u),\qquad
\sup_{i\in\cI}\tilde\psi^*(\rho_i\|\sigma_i|1)=\lim_{u\nearrow1}\tilde\psi^*(\rho\|\sigma|1).
\end{align}
From \eqref{F-3.11}--\eqref{F-3.13} we conclude that
\begin{align*}
sc_r^0(\rho\|\sigma)
&\le\sup_{u\in(0,1)}\Bigl\{ur-\tilde\psi^*(\rho\|\sigma|u)\Bigr\} \\
&=\sup_{\alpha>1}{\alpha-1\over\alpha}\bigl\{r-D_\alpha^*(\rho\|\sigma)\bigr\}=H_r^*(\rho\|\sigma),
\end{align*}
proving \eqref{F-3.10}.
\end{proof}

\begin{remark}\label{R-3.8}\rm
According to \cite[Example III.39]{Mo} (also \cite[Remark 5.4\,(1)]{Hi3}), there exist commuting
density operators $\rho,\sigma$ on $\cH$ with $\dim\cH=+\infty$ such that
\[
D(\rho\|\sigma)<+\infty,\qquad D_\alpha^*(\rho\|\sigma)=+\infty\quad
\mbox{for all $\alpha\in(1,+\infty)$}.
\]
In this case, $H_r^*(\rho\|\sigma)=-\infty$ for all $r\in\bR$, whereas any exponent in \eqref{F-3.5} is
non-negative. See \cite[Theorem IV.5]{Mo} for a more precise result in this case. Therefore, the
assumption of $D_\alpha^*(\rho\|\sigma)<+\infty$ for some $\alpha>1$ is essential in Theorem
\ref{T-3.7}, as well as in the convergence \eqref{F-2.5}.
\end{remark}

For given states $\rho,\sigma\in\M_*^+$ and any $r\ge0$, the $n$th minimal type I error probability
of Hoeffding type is defined as
\begin{align}\label{F-3.14}
\alpha_{e^{-nr}}^*(\rho_n\|\sigma_n)
:=\min_{0\le T_n\le1}\bigl\{\rho_n(1-T_n):\sigma_n(T_n)\le e^{-nr}\bigr\},
\end{align}
where the minimum is taken over all tests in $\M^{\overline\otimes n}$ with $\sigma_n(T_n)\le e^{-nr}$.
Note that the minimum exists since the set of such tests is weakly compact. The $n$th maximal type
I success probability is then given as
\[
1-\alpha_{e^{-nr}}^*(\rho_n\|\sigma_n)
=\max_{0\le T_n\le1}\bigl\{\rho_n(T_n):\sigma_n(T_n)\le e^{-nr}\bigr\}.
\]
In terms of this we can reformulate the assertion of Theorem \ref{T-3.7} on the strong converse
exponents as follows:

\begin{theorem}\label{T-3.9}
Under the same assumption as in Theorem \ref{T-3.7}, for every $r\ge0$ we have
\[
\lim_{n\to\infty}-{1\over n}\log\bigl\{1-\alpha_{e^{-nr}}^*(\rho_n\|\sigma_n)\bigr\}
=H_r^*(\rho\|\sigma).
\]
\end{theorem}

\begin{proof}
For each $n$ one can choose a test $T_n$ in $\M^{\overline\otimes n}$ such that
$\sigma_n(T_n)\le e^{-nr}$ and $\rho_n(T_n)=1-\alpha_{e^{-nr}}^*(\rho_n\|\sigma_n)$. Then it follows
that
\[
\liminf_{n\to\infty}-{1\over n}\log\bigl\{1-\alpha_{e^{-nr}}^*(\rho_n\|\sigma_n)\bigr\}
=\liminf_{n\to\infty}-{1\over n}\log\rho_n(T_n)\ge\underline{sc}_r(\rho\|\sigma).
\]
On the other hand, if a sequence of tests $\{T_n\}$ satisfies $\liminf_n-{1\over n}\log\sigma_n(T_n)>r$,
then we have $1-\alpha_{e^{-nr}}^*(\rho_n\|\sigma_n)\ge\rho_n(T_n)$ for all sufficiently large $n$ and
hence
\[
\limsup_{n\to\infty}-{1\over n}\log\bigl\{1-\alpha_{e^{-nr}}^*(\rho_n\|\sigma_n)\bigr\}
\le\limsup_{n\to\infty}-{1\over n}\log\rho_n(T_n).
\]
Therefore,
\[
\limsup_{n\to\infty}-{1\over n}\log\bigl\{1-\alpha_{e^{-nr}}^*(\rho_n\|\sigma_n)\bigr\}
\le\overline{sc}_r^0(\rho\|\sigma).
\]
Since $\underline{sc}_r(\rho\|\sigma)=\overline{sc}_r^0(\rho\|\sigma)=H_r^*(\rho\|\sigma)$ by Theorem
\ref{T-3.7}, we find that
\[
\lim_{n\to\infty}-{1\over n}\log\bigl\{1-\alpha_{e^{-nr}}^*(\rho_n\|\sigma_n)\bigr\}
=H_r^*(\rho\|\sigma),
\]
as asserted.
\end{proof}

Using Theorem \ref{T-3.7} we can give a direct operational interpretation of the sandwiched 
R\'enyi divergences as generalized cutoff rates, following the idea of Csisz\'ar \cite{Csiszar}.

\begin{definition}\label{D-3.10}\rm
Let $\rho,\sigma\in\M_*^+$ be states, and $\kappa\in(0,1)$. The
\emph{generalized $\kappa$-cutoff rate} $C_{\kappa}(\rho\|\sigma)$ is defined to be the infimum
of all $r_0\in\bR$ such that $\scli_r(\rho\|\sigma)\ge \kappa(r-r_0)$ holds for every $r\in\bR$. 
\end{definition}

\begin{theorem}\label{T-3.11}
Assume that $\M$ is injective. Let $\rho,\sigma\in\M_*^+$ be states such that
$D_{\alpha_0}^*(\rho\|\sigma)<+\infty$ for some $\alpha_0>1$. Then 
\begin{align*}
C_{\kappa}(\rho\|\sigma)= D_{\frac{1}{1-\kappa}}\nw(\rho\|\sigma),
\quad\text{or equivalently,}\quad
D_{\alpha}\nw(\rho\|\sigma)=C_{\frac{\alpha-1}{\alpha}}(\rho\|\sigma)
\end{align*}
for every $\alpha\in(1,\alpha_0)$ and corresponding 
$\kappa=(\alpha-1)/\alpha\in(0,\kappa_0)$, where $\kappa_0:=(\alpha_0-1)/\alpha_0$.
\end{theorem}

\begin{proof}
By Theorem \ref{T-3.7} and \eqref{F-3.6},
\begin{align*}
\scli_r(\rho\|\sigma)
&=
H_r\nw(\rho\|\sigma)
=
\sup_{u\in(0,1)}\{u r-\tilde\psi\nw(\rho\|\sigma|u)\} \\
&\ge
\kappa r-\tilde\psi\nw(\rho\|\sigma|\kappa)
=\kappa\Bigl(r-D_{\frac{1}{1-\kappa}}\nw(\rho\|\sigma)\Bigr),
\end{align*}
showing $C_{\kappa}(\rho\|\sigma)\ge D_{\frac{1}{1-\kappa}}\nw(\rho\|\sigma)$.
As we have seen in the proof of Theorem \ref{T-3.7}, 
$\tilde\psi\nw(\rho_i\|\sigma_i|\cdot)$ is convex and continuous on $(0,1)$ for every $i$,
and hence $\tilde\psi\nw(\rho\|\sigma|\cdot)=\sup_i\tilde\psi\nw(\rho_i\|\sigma_i|\cdot)$
is convex and lower semi-continuous on $(0,1)$, where the equality is by Theorem \ref{T-3.1}.
By assumption, and the monotonicity of $D_{\alpha}\nw$ in $\alpha$ stated in Section \ref{Sec-2}
(see \cite[Proposition 3.7]{Je1}), $\tilde\psi\nw(\rho\|\sigma|\cdot)$ is finite-valued on
$(0,\kappa_0)$, and hence it has finite left- and right-derivatives at every
$\kappa\in(0,\kappa_0)$.  Thus, for any such $\kappa$ and
$r\in [\derleft\tilde\psi\nw(\rho\|\sigma|\kappa),\derright\tilde\psi\nw(\rho\|\sigma|\kappa)]$, 
\begin{align*}
\scls_r(\rho\|\sigma)&=
H_r\nw(\rho\|\sigma)=\sup_{u\in(0,1)}\{u r-\tilde\psi\nw(\rho\|\sigma|u)\} \\
&=
\kappa r- \tilde\psi\nw(\rho\|\sigma|\kappa)
=\kappa\Bigl(r-D_{\frac{1}{1-\kappa}}\nw(\rho\|\sigma)\Bigr),
\end{align*}
where the first equality is again due to Theorem \ref{T-3.7}. This shows
$C_{\kappa}(\rho\|\sigma)\le D_{\frac{1}{1-\kappa}}\nw(\rho\|\sigma)$, completing the proof.
\end{proof}

In the rest of this section we consider the (regularized) measured R\'enyi divergences.
Let us first recall these notions in the present setting. A measurement (POVM) in $\M$ is given by
a finite family $\fM=(M_j)_{1\le j\le k}$ of $M_j\in\M_+$ ($j=1,\dots,k$) such that
$\sum_{j=1}^kM_j=1$. For each $\rho\in\M_*^+$ let $\fM(\rho):=(\rho(M_j))_{1\le j\le k}$
denote the post-measurement probability distribution on $\{1,\dots,k\}$.

\begin{definition}\label{D-3.12}\rm
The \emph{measured R\'enyi divergence}
$D_\alpha^\meas(\rho\|\sigma)$ and its \emph{regularized version}
$\overline{D}_\alpha^\meas(\rho\|\sigma)$ for $\alpha\in[1/2,+\infty)\setminus\{1\}$ are defined as
\begin{align*}
D_\alpha^\meas(\rho\|\sigma)&:=\sup\{D_\alpha(\fM(\rho)\|\fM(\sigma)):
\mbox{$\fM$ a measurement in $M$\}}, \\
\overline{D}_\alpha^\meas(\rho\|\sigma)
&:=\sup_{n\in\bN}{1\over n}D_\alpha^\meas(\rho^{\otimes n}\|\sigma^{\otimes n})
=\lim_{n\to\infty}{1\over n}D_\alpha^\meas(\rho^{\otimes n}\|\sigma^{\otimes n}).
\end{align*}
The last equality above holds since
$n\mapsto D_\alpha^\meas(\rho^{\otimes n}\|\sigma^{\otimes n})$ is superadditive, as immediately
seen from the definition of $D_\alpha^\meas$.
\end{definition}

One might also consider more general notions of measurements in the definition, but that does not
change the value of the (regularized) measured R\'enyi divergence, as was shown in 
\cite[Proposition 5.2]{Hi}. In the opposite direction, one might restrict measurements $\fM$ to 
two-valued ones (tests), which leads to the notion of the \emph{test-measured R\'enyi divergence}
$D_\alpha^\test(\rho\|\sigma)$ and its regularized version $\overline D_\alpha^\test(\rho\|\sigma)$,
defined as
\begin{align*}
D_\alpha^\test(\rho\|\sigma)&:=\sup_{T\in\M,\,0\le T\le1}
D_\alpha((\rho(T),\rho(1-T))\|(\sigma(T),\sigma(1-T))), \\
\overline{D}_\alpha^\test(\rho\|\sigma)
&:=\sup_{n\in\bN}{1\over n}D_\alpha^\test(\rho^{\otimes n}\|\sigma^{\otimes n}).
\end{align*}
It is obvious that
\begin{align}\label{F-3.15}
\begin{array}{lllll}
D_\alpha^\test(\rho\|\sigma) &\le & \overline{D}_\alpha^\test(\rho\|\sigma) \\
\ds\ds\vertleq & & \ds\ds\vertleq \\
D_\alpha^\meas(\rho\|\sigma) & \le & \overline{D}_\alpha^\meas(\rho\|\sigma)
\end{array}
\end{align}
for all $\alpha\in[1/2,+\infty)\setminus\{1\}$.

Let $\M$ be injective, and $\{\M_i\}_{i\in\cI}$ be an increasing net of finite-dimensional
*-subalgebras of $\M_0:=e\M e$ with $\M_0=\bigl(\bigcup_i\M_i\bigr)''$, where
$e:=s(\rho)\vee s(\sigma)$ (see the proof of Theorem \ref{T-3.7}). 
Note that for each $n\in\bN$, $\M_0^{\overline\otimes n}=\bigl(\bigcup_i\M_i^{\otimes n}\bigr)''$,
$(\rho|_{\M_0})^{\otimes n}=\rho^{\otimes n}|_{\M_0^{\overline\otimes n}}$, and similarly for
$\sigma$. Hence, for every $\alpha\in[1/2,+\infty)\setminus\{1\}$ and every $n\in\bN$ one has by
Theorem \ref{T-3.1},
\begin{align}
D_\alpha^*(\rho^{\otimes n}\|\sigma^{\otimes n})
&=\lim_iD_\alpha^*\bigl((\rho|_{\M_i})^{\otimes n}\|(\sigma|_{\M_i})^{\otimes n}\bigr)
\nonumber\\
&=\lim_inD_\alpha^*(\rho|_{\M_i}\|\sigma|_{\M_i})=nD_\alpha^*(\rho\|\sigma),
\label{F-3.16}
\end{align}
where the second equality above follows from the additivity of $D_\alpha^*$ under tensor
product in the finite-dimensional case. Indeed, the additivity of $D_\alpha^*$ under tensor
product holds true in the general von Neumann algebra case, as observed in \cite{BST} (see
the footnote of p.~1860) in the approach of Araki and Masuda's $L^p$-norms.

\begin{proposition}\label{P-3.13}
Assume that $\M$ is injective. Then for every $\rho,\sigma\in\M_*^+$ and every
$\alpha\in[1/2,+\infty)\setminus\{1\}$ we have
\begin{align}\label{F-3.17}
D_\alpha^*(\rho\|\sigma)
&=\overline{D}_\alpha^\meas(\rho\|\sigma)
=\lim_{n\to\infty}{1\over n}D_\alpha^\meas(\rho^{\otimes n}\|\sigma^{\otimes n}),
\end{align}
and for every $\alpha>1$,
\begin{align}\label{F-3.18}
D_\alpha^*(\rho\|\sigma)
&=\overline{D}_\alpha^\test(\rho\|\sigma)
=\lim_{n\to\infty}{1\over n}D_\alpha^\test(\rho^{\otimes n}\|\sigma^{\otimes n}).
\end{align}
\end{proposition}

\begin{proof}
For any $\alpha\in[1/2,+\infty)\setminus\{1\}$, by \eqref{F-3.16} and \eqref{F-3.15} note that
\[
D_\alpha^*(\rho\|\sigma)\ge\overline D_\alpha^\meas(\rho\|\sigma)
\ge\overline D_\alpha^\test(\rho\|\sigma).
\]
To prove \eqref{F-3.17}, it suffices to show that
\begin{align}\label{F-3.19}
\liminf_{n\to\infty}{1\over n}D_\alpha^\meas(\rho^{\otimes n}\|\sigma^{\otimes n})
\ge D_\alpha^*(\rho\|\sigma).
\end{align}
Let $e:=s(\rho)\vee s(\sigma)$. Let $v<D_\alpha^*(\rho\|\sigma)$ be arbitrary. By Theorem
\ref{T-3.1} there exists a finite-dimensional *-subalgebra $\N_0$ of $e\M e$ such that
$D_\alpha^*(\rho|_{\N_0}\|\sigma|_{\N_0})>v$. Write $\rho_0:=\rho|_{\N_0}$ and
$\sigma_0:=\sigma|_{\N_0}$. From \cite[Theorem 3.7]{MO1} and \cite[Corollary 4]{HT} it
follows that
\[
D_\alpha^*(\rho_0\|\sigma_0)
=\lim_{n\to\infty}{1\over n}D_\alpha^\meas(\rho_0^{\otimes n}\|\sigma_0^{\otimes n}).
\]
Hence there exists an $n_0\in\bN$ such that
${1\over n}D_\alpha^\meas(\rho_0^{\otimes n}\|\sigma_0^{\otimes n})>v$ for all $n\ge n_0$.
Therefore,
\[
\liminf_{n\to\infty}{1\over n}D_\alpha^\meas(\rho^{\otimes n}\|\sigma^{\otimes n})
\ge\liminf_{n\to\infty}{1\over n}D_\alpha^\meas(\rho_0^{\otimes n}\|\sigma_0^{\otimes n})
\ge v,
\]
which implies \eqref{F-3.19} by letting $v\nearrow D_\alpha^*(\rho\|\sigma)$.

When $\alpha>1$, the proof of \eqref{F-3.18} can proceed in the same way as above by
appealing to \cite[Corollary 4.6]{MO1}.
\end{proof}

\begin{remark}\label{R-3.14}\rm
When $\alpha<1$, it may happen that 
$\overline{D}_\alpha^\test(\rho\|\sigma)<\overline{D}_\alpha^\meas(\rho\|\sigma)$, even for
finite-dimensional commuting states $\rho$ and $\sigma$; see \cite{testdiv}.
\end{remark}

\section{The $C^*$-algebra case}\label{Sec-4}

\subsection{Sandwiched and standard R\'enyi divergences in $C^*$-algebras}\label{Sec-4.1}

In this section we first extend the notion of the sandwiched and the standard R\'enyi
divergences to positive linear functionals on a general unital $C^*$-algebra. In the rest of the
section let $\cA$ be a unital $C^*$-algebra. (The unitality assumption is not essential, since we
can work with the unitization $\cA\oplus\bC$ if $\cA$ is non-unital.) Let $\cA_+^*$ be the set of
positive linear functionals (automatically bounded) on $\cA$. To define $D_\alpha^*$ and
$D_\alpha$ for $\rho,\sigma\in\cA_+^*$, let us consider the universal representation
$\{\pi_u,\cH_u\}$ of $\cA$; then $\rho$ and $\sigma$ have the respective normal extensions
$\overline\rho$ and $\overline\sigma$ to $\pi_u(\cA)''\cong\cA^{**}$ such that
$\rho=\overline\rho\circ\pi$ and $\sigma=\overline\sigma\circ\pi$; see, e.g.,
\cite[Definition III.2.3, Theorem III.2.4]{Ta1}.

\begin{definition}\label{D-4.1}\rm
For every $\rho,\sigma\in\cA_+^*$ let $\overline\rho,\overline\sigma$ be as stated above. For every
$\alpha\in[1/2,+\infty)\setminus\{1\}$ we define the \emph{sandwiched R\'enyi $\alpha$-divergence}
of $\rho$ and $\sigma$ by
\begin{align}\label{F-4.1}
\widehat D_\alpha^*(\rho\|\sigma):=D_\alpha^*(\overline\rho\|\overline\sigma).
\end{align}
For $\alpha\in[0,+\infty)\setminus\{1\}$ define also the \emph{standard R\'enyi $\alpha$-divergence}
of $\rho$ and $\sigma$ by
\begin{align}\label{F-4.2}
\widehat D_\alpha(\rho\|\sigma):=D_\alpha(\overline\rho\|\overline\sigma).
\end{align}
\end{definition}

One question arises immediately: when $\cA$ itself is a von Neumann algebra and $\rho,\sigma$
are normal functionals, do definitions \eqref{F-4.1} and \eqref{F-4.2} give the same notions of
sandwiched and standard R\'enyi divergences (see Section \ref{Sec-2} for the definitions),
i.e., do the equalities
\begin{align}\label{F-4.3}
\widehat D_\alpha^*(\rho\|\sigma)=D_\alpha^{*}(\rho\|\sigma),\qquad
\widehat D_\alpha(\rho\|\sigma)=D_\alpha(\rho\|\sigma)
\end{align}
hold? Theorem \ref{T-4.3} below shows that these are indeed the case. For stating it,  let us
introduce the following:

\begin{definition}\label{D-4.2}\rm
Let $\rho,\sigma\in\cA_+^*$. We say that a representation $\pi:\,\cA\to\B(\hil)$ is 
\emph{$(\rho,\sigma)$-normal}, if there exist $\rho_{\pi},\sigma_{\pi}\in(\pi(\cA)'')_*^{+}$ such that 
$\rho=\rho_{\pi}\circ\pi$, $\sigma=\sigma_{\pi}\circ\pi$.
\end{definition}

\begin{theorem}\label{T-4.3}
Let $\rho,\sigma\in\cA_+^*$ and $\pi$ be any $(\rho,\sigma)$-normal representation of $\cA$.
Then we have the following:
\begin{itemize}
\item[\rm(i)] For any $\alpha\in[1/2,+\infty)\setminus\{1\}$,
\begin{align}\label{F-4.4}
\widehat D_{\alpha}^*(\rho\|\sigma)=D_{\alpha}^*(\rho_{\pi}\|\sigma_{\pi}),\quad\mbox{i.e.},\quad
D_{\alpha}^*(\overline\rho\|\overline\sigma)=D_{\alpha}^*(\rho_{\pi}\|\sigma_{\pi}).
\end{align}
\item[\rm(ii)] For any $\alpha\in[0,+\infty)\setminus\{1\}$,
\begin{align}\label{F-4.5}
\widehat D_{\alpha}(\rho\|\sigma)=D_{\alpha}(\rho_{\pi}\|\sigma_{\pi}),\quad\mbox{i.e.},\quad
D_{\alpha}(\overline\rho\|\overline\sigma)=D_{\alpha}(\rho_{\pi}\|\sigma_{\pi}).
\end{align}
\end{itemize}

In particular, the equalities in \eqref{F-4.3} hold when $\cA$ is a von Neumann
algebra and $\rho,\sigma$ are normal functionals.
\end{theorem}

The proof of the theorem is somewhat technical based on Kosaki's interpolation $L^p$-spaces
\cite{Ko1}, so we defer it to Appendix \ref{Sec-H}. Once \eqref{F-4.3} is
confirmed by Theorem \ref{T-4.3}, we may and do rewrite $\widehat D_\alpha^*(\rho\|\sigma)$
and $\widehat D(\rho\|\sigma)$ for $\rho,\sigma\in\cA_+^*$ as $D_\alpha^*(\rho\|\sigma)$ and
$D_\alpha(\rho\|\sigma)$, respectively.

\begin{remark}\label{R-4.4}\rm
Recall that $D_{1/2}^*(\rho\|\sigma)=-2\log F(\rho,\sigma)$ for $\rho,\sigma\in\M_*^+$, where
$\M$ is a von Neumann algebra and $F(\rho,\sigma)$ is the \emph{fidelity} of $\rho,\sigma$.
Theorem \ref{T-4.3} (for $\alpha=1/2$) shows that the fidelity of states $\rho,\sigma\in\cA_+^*$
may be defined by $F(\rho_{\pi},\sigma_{\pi})$ via any $(\rho,\sigma)$-normal representation
$\pi$ of $\cA$. In view of \eqref{F-2.4}, Theorem \ref{T-4.3} also shows that the relative
entropy $D(\rho\|\sigma)$ of $\rho,\sigma\in\cA_+^*$ can be defined by
$D(\rho_\pi\|\sigma_\pi)$ as above; see \cite[Sec.~5]{Ar2} and \cite[Lemma 3.1]{HOT}.
\end{remark}

Based on Definition \ref{D-4.1} and Theorem \ref{T-4.3} we can easily extend properties of
the sandwiched and the standard R\'enyi divergences from the von Neumann setting to the
$C^*$-algebra setting. For instance, some important ones are given in the next proposition.

\begin{proposition}\label{P-4.5}
Let $\cA$ be a unital $C^*$-algebra and $\rho,\sigma\in\cA_+^*$.
\begin{itemize}
\item[\rm(1)] \emph{Joint lower semi-continuity:}
The map $(\rho,\sigma)\in\cA_+^*\times\cA_+^*\mapsto D_\alpha^*(\rho\|\sigma)$ is jointly lower
semi-continuous in the norm topology for every $\alpha\in(1,+\infty)$ and jointly continuous in the
norm topology for every $\alpha\in[1/2,1)$. The map
$(\rho,\sigma)\in\cA_+^*\times\cA_+^*\mapsto D_\alpha(\rho\|\sigma)$ is jointly lower
semi-continuous in the weak (i.e., $\sigma(\cA^*,\cA^{**})$-) topology for every
$\alpha\in(1,2]$ and jointly continuous in the norm topology for every $\alpha\in[0,1)$.

\item[\rm(2)] \emph{Monotonicity (Data-processing inequality):} Let $\Phi:\B\to\cA$ be a
unital positive linear map between unital $C^*$-algebras. Then for every
$\alpha\in[1/2,+\infty)\setminus\{1\}$,
\[
D_\alpha^*(\rho\circ\Phi\|\sigma\circ\Phi)\le D_\alpha^*(\rho\|\sigma).
\]
If $\Phi$ is, in addition, a Schwarz map, then for every $\alpha\in[0,2]\setminus\{1\}$,
\[
D_\alpha(\rho\circ\Phi\|\sigma\circ\Phi)\le D_\alpha(\rho\|\sigma).
\]

\item[\rm(3)] \emph{Inequality between $D_\alpha^*$ and $D_\alpha$:} For every
$\alpha\in[1/2,+\infty)\setminus\{1\}$,
\[
D_\alpha^*(\rho\|\sigma)\le D_\alpha(\rho\|\sigma).
\]

\item[\rm(4)] \emph{Martingale convergence:} Assume that $\{\cA_i\}_{i\in\cI}$ is an increasing
net of $C^*$-subalgebras of $\cA$ containing the unit of $\cA$ such that $\bigcup_{i\in\cI}\cA_i$
is norm-dense in $\cA$. Then for every $\alpha\in[0,2]\setminus\{1\}$,
\[
D_\alpha(\rho|_{\cA_i}\|\sigma|_{\cA_i})\,\nearrow\,D_\alpha(\rho\|\sigma).
\]
Moreover, if there exists a $(\rho,\sigma)$-normal representation $\pi$ of $\cA$ such that
$\pi(\cA)''$ is $\sigma$-finite (this is automatic if $\cA$ is separable), then for every
$\alpha\in[1/2,+\infty)\setminus\{1\}$,
\[
D_\alpha^*(\rho|_{\cA_i}\|\sigma|_{\cA_i})\,\nearrow\,D_\alpha^*(\rho\|\sigma).
\]
\end{itemize}
\end{proposition}

\begin{proof}
(1)\enspace
Let $\rho_n,\sigma_n\in\cA_+^*$, $n\in\bN$, be such that $\|\rho_n-\rho\|\to0$ and
$\|\sigma_n-\sigma\|\to0$, with normal extensions $\overline\rho_n,\overline\sigma_n$ to
$\pi_u(\cA)''$ as well as $\overline\rho,\overline\sigma$.
By Kaplansky's density theorem, note that $\|\overline\rho_n-\overline\rho\|=\|\rho_n-\rho\|$ and
$\|\overline\sigma_n-\overline\sigma\|=\|\sigma_n-\sigma\|$. Hence the assertions follow from
the corresponding result in the von Neumann algebra case; see \cite[Proposition 3.10]{Je1} and
\cite[Theorem 3.16\,(3)]{Hi} for $D_\alpha^*$, and \cite[Proposition 5.3\,(6)]{Hi3} and
\cite[Corollary 3.8]{Hi} for $D_\alpha$.

(2)\enspace
As shown in the proof of \cite[Proposition 7.4]{Hi4}, $\Phi$ extends to a unital positive
normal map $\overline\Phi:\pi_u(\B)''\to\pi_u(\cA)''$ in such a way that
$\overline\Phi\circ\pi_u=\pi_u\circ\Phi$, where $\pi_u$ on the left-hand side is the
universal representation of $\B$ and that of $\cA$ is on the right-hand side. Then
$\overline\rho\circ\overline\Phi$ and $\overline\sigma\circ\overline\Phi$ are the normal
extensions of $\rho\circ\Phi$ and $\sigma\circ\Phi$ to $\pi_u(\B)''$, as seen in the
proof of \cite[Proposition 7.4]{Hi4}. By \cite[Theorem 3.14]{Je1} and
\cite[Theorem 4.1]{Je2}\footnote{
In \cite{Je1,Je2} (also \cite{Hi}) it was implicitly assumed that $\M$ is $\sigma$-finite. But
this assumption can be removed. Indeed, let $\Phi:\N\to\M$ be a unital positive normal
map between von Neumann algebras and $\rho,\sigma\in\M_*^+$. Let
$e:=s(\rho)\vee s(\sigma)$ and $e_0:=s(\rho\circ\Phi)\vee s(\sigma\circ\Phi)$. Since
$s(\rho)(1-\Phi(e_0))s(\rho)=0$, $s(\rho)=s(\rho)\Phi(e_0)=\Phi(e_0)s(\rho)$ and
$s(\rho)\le\Phi(e_0)$. Similarly, $s(\sigma)=s(\sigma)\Phi(e_0)=\Phi(e_0)s(\sigma)$
and $s(\sigma)\le\Phi(e_0)$. Let $P:=1-s(\rho)$, $Q:=1-s(\sigma)$ and
$A:=1-\Phi(e_0)\ge0$; hence $P\ge A$ and $Q\ge A$. Note that $PQP\ge PAP=A$,
$QPQPQ\ge QAQ=A$, and so on. Hence $(PQ)^nP\ge A$ for all
$n\ge1$. Since $(PQ)^nP\to P\wedge Q$ strongly, we have $P\wedge Q\ge A$, which means
that $e\le\Phi(e_0)$. We thus find that $\tilde\Phi:=e\Phi(\cdot)e|_{e_0\N e_0}$
is a unital positive normal map from $e_0\N e_0$ to $e\M e$. Note that $e_0\N e_0$ and
$e\M e$ are $\sigma$-finite. Moreover,
$D_\alpha^*(\rho\|\sigma)=D_\alpha^*(\rho|_{e\M e}\|\sigma|_{e\M e})$ and
\[
D_\alpha^*(\rho\circ\Phi\|\sigma\circ\Phi)
=D_\alpha^*(\rho\circ\Phi|_{e_0\N e_0}\|\sigma\circ\Phi|_{e_0\N e_0})
=D_\alpha^*((\rho|_{e\M e})\circ\tilde\Phi\|(\sigma|_{e\M e})\circ\tilde\Phi).
\]}
we have $D_\alpha^*(\rho\circ\Phi\|\sigma\circ\Phi)\le D_\alpha^*(\rho\|\sigma)$ for all
$\alpha\in[1/2,+\infty)\setminus\{1\}$. In addition, assume that $\Phi$ is a Schwarz map.
Then it is easy to see by Kaplansky's density theorem that $\overline\Phi$ is a Schwarz
map again. Hence \cite[Proposition 5.3\,(9)]{Hi3} gives
$D_\alpha(\rho\circ\Phi\|\sigma\circ\Phi)\le D_\alpha(\rho\|\sigma)$ for all
$\alpha\in[0,2]\setminus\{1\}$.

(3)\enspace
By \cite[Theorem 12]{BST} (also \cite[Corollary 3.6]{Je1}) we have
\[
D_\alpha^*(\rho\|\sigma)=D_\alpha^*(\overline\rho\|\overline\sigma)
\le D_\alpha(\overline\rho\|\overline\sigma)=D_\alpha(\rho\|\sigma).
\]

(4)\enspace
Let $\pi$ be as stated in the latter assertion (for $D_\alpha^*$). Then it is clear that
$\{\pi(\cA_i)''\}_{i\in\cI}$ is an increasing net of von Neumann subalgebras of $\sigma$-finite
$\pi(\cA)''$ containing the unit of $\pi(\cA)''$ such that $\bigl(\bigcup_i\pi(\cA_i)''\bigr)''=\pi(\cA)''$.
By Theorems \ref{T-4.3} and \ref{T-3.1}, for every $\alpha\in[1/2,+\infty)\setminus\{1\}$ we have
\[
D_\alpha^*(\rho|_{\cA_i}\|\sigma|_{\cA_i})
=D_\alpha^*(\rho_{\pi}|_{\pi(\cA_i)''}\|\sigma_{\pi}|_{\pi(\cA_i)''})
\,\nearrow\,D_\alpha^*(\rho_{\pi}\|\sigma_{\pi})=D_\alpha^*(\rho\|\sigma).
\]
As for the first assertion, note that $Q_\alpha(\overline\rho\|\overline\sigma)$ for $\alpha\in(1,2]$
is the standard $f$-divergence of $\overline\rho,\overline\sigma$ for $f(t)=t^\alpha$ and
$-Q_\alpha(\overline\rho\|\overline\sigma)$ for $\alpha\in[0,1)$ is that for $f(t)=-t^\alpha$. Since
these functions $f_\alpha$ are operator convex on $(0,+\infty)$, the martingale convergence for
$D_\alpha$ when $\alpha\in[0,2]\setminus\{1\}$ follows from \cite[Theorem 4.1\,(v)]{Hi3}.
\end{proof}

\begin{example}\label{E-4.6}\rm
Let $\cC(\cH)$ be the compact operator ideal in $\B(\cH)$, and $\cC_1(\cH)$ be the trace class
on $\cH$. Consider the unital $C^*$-algebra $\cA:=\cC(\cH)+\bC\1$. It is well known that
$\cA^*\cong\cC_1(\cH)$ and $\cC_1(\cH)^*\cong\B(\cH)$. Thus we can identify
$\cA_+^*=\cC_1(\cH)_+=\B(\cH)_*^+$. With this identification, for every
$\rho,\sigma\in\cC_1(\cH)_+$, the $C^*$-versions $D_\alpha^*(\rho\|\sigma)$ and
$D_\alpha(\rho\|\sigma)$ of $\rho,\sigma$ considered as elements of $\cA_+^*$ coincide with
those of $\rho,\sigma$ as elements of $\B(\cH)_*^+$. Recently in \cite{Mo}, the
$(\alpha,z)$-R\'enyi divergences, a more general notion than the sandwiched and the standard
R\'enyi divergences, have been discussed for more general operators
$\rho,\sigma$ in $\B(\cH)_+$, as well as their strong converse exponents.
\end{example}

\begin{example}\label{E-4.7}\rm
Let $(\cA,G,\tau)$ be a $C^*$-dynamical system, where $G$ is a group and $\tau$ is an action of
$G$ as *-automorphisms of $\cA$. Let $\cS_G(\cA)$ denote the set of $\tau$-invariant states of
$\cA$. Let $\ffi\in\cS_G(\cA)$ and $\{\pi_\ffi,\cH_\ffi,\xi_\ffi\}$ be the cyclic representation of $\cA$
induced by $\ffi$. Then there is a unique unitary representation $U_\ffi$ of $G$ on $\cH_\ffi$ such
that $U_\ffi(g)\xi_\ffi=\xi_\ffi$ and $\pi_\ffi(\tau_g(a))=U_\ffi(g)\pi_\ffi(a)U_\ffi(g)^*$ for all $a\in\cA$,
$g\in G$. It is well known (see, e.g., \cite[Corollary 4.3.11]{BR}) that $\cA$ is \emph{$G$-abelian},
i.e., $\pi_\ffi(\cA)'\cap U_\ffi(G)'$ is abelian for all $\ffi\in\cS_G(\cA)$ if and only if $\cS_G(\cA)$ is a
(Choquet) simplex. In this case, each $\ffi\in\cS_G(\cA)$ has a unique maximal representing
measure $\mu_\ffi$ on $\cS_G(\cA)$, which is the $\pi_\ffi(\cA)'\cap U_\ffi(G)'$-orthogonal measure
of $\ffi$; see \cite[Proposition 4.3.3]{BR}. Define a unital positive linear map
$\Phi:\cA\to C(\cS_G(\cA))$ (consisting of all continuous functions on $\cS_G(\cA)$) by
$(\Phi a)(\omega):=\omega(a)$ for $a\in\cA$, $\omega\in\cS_G(\cA)$. When $\cA$ is
$G$-abelian and $\rho,\sigma\in\cS_G(\cA)$, since $\rho=\mu_\rho\circ\Phi$ and
$\sigma=\mu_\sigma\circ\Phi$, Proposition \ref{P-4.5}\,(2) gives
\begin{equation}\label{F-4.6}
\begin{aligned}
D_\alpha^*(\rho\|\sigma)&\le D_\alpha(\mu_\rho\|\mu_\sigma),\qquad
\alpha\in[1/2,+\infty)\setminus\{1\}, \\
D_\alpha(\rho\|\sigma)&\le D_\alpha(\mu_\rho\|\mu_\sigma),\qquad
\alpha\in[0,2]\setminus\{1\},
\end{aligned}
\end{equation}
where $D_\alpha(\mu_\rho\|\mu_\sigma)$ is the classical R\'enyi relative entropy. Now assume
that $\cA$ is \emph{$G$-central}, i.e., $\pi_\ffi(\cA)'\cap U_\ffi(G)'\subseteq\pi_\ffi(\cA)''$ (hence
$\pi_\ffi(\cA)'\cap U_\ffi(G)'\subseteq$ the center of $\pi_\ffi(\cA)''$) for all $\ffi\in\cS_G(\cA)$, which
is stronger than $G$-abeliannes and weaker than some other conditions of asymptotic abeliannes;
see \cite[Sec.~4.3]{BR}, \cite{DKS}. In this case, with $\ffi:=(\rho+\sigma)/2$ there is an
isomorphism $\theta_\ffi:L^\infty(\cS_G(\cA),\mu_\ffi)\to\pi_\ffi(\cA)'\cap U_\ffi(G)'$ such that
\begin{align}\label{F-4.7}
\<\xi_\ffi,\theta_\ffi(f)\pi_\ffi(a)\xi_\ffi\>
=\int f(\omega)\omega(a)\,d\mu_\ffi(\omega),\qquad f\in L^\infty(\cS_G(\cA),\mu_\ffi).
\end{align}
For this, see \cite[Proposition 41.22]{BR}. Since $\rho,\sigma\le2\ffi$, note \cite[Corollary 4.1.17]{BR}
that $\mu_\rho,\mu_\sigma\le2\mu_\ffi$ so that one can take $g_\rho:=d\mu_\rho/d\mu_\ffi$ and
$g_\sigma:=d\mu_\sigma/d\mu_\ffi$ in $L^\infty(\cS_G(\cA),\mu_\ffi)_+$. Then for every $a\in\cA$
one has
\[
\rho(a)=\int g_\rho(\omega)\omega(a)\,d\mu_\ffi(\omega)
=\<\xi_\ffi,\theta_\ffi(g_\rho)\pi_\ffi(a)\xi_\ffi\>
\]
thanks to \eqref{F-4.7}. Hence the normal extension of $\rho$ to $\pi_\ffi(\cA)''$ is given as
\[
\tilde\rho(x)=\<\xi_\ffi,\theta_\ffi(g_\rho)x\xi_\ffi\>,\qquad x\in\pi_\ffi(\cA)'',
\]
which implies that
\begin{align*}
\tilde\rho(\theta_\ffi(f))&=\<\xi_\ffi,\theta_\ffi(g_\rho f)\xi_\ffi\>
=\int g_\rho(\omega)f(\omega)\,d\mu_\ffi(\omega)\quad\mbox{(by \eqref{F-4.7})} \\
&=\int f(\omega)\,d\mu_\rho(\omega),\qquad f\in L^\infty(\cS_G(\cA),\mu_\ffi).
\end{align*}
Hence $\tilde\rho\circ\theta_\ffi=\mu_\rho$ follows. Similarly, replacing $\rho,g_\rho$ with
$\sigma,g_\sigma$ in the above argument, one has $\tilde\sigma\circ\theta_\ffi=\mu_\sigma$ for
the normal extension $\tilde\sigma$ of $\sigma$ to $\pi_\ffi(\cA)''$. From Proposition \ref{P-4.5}\,(2)
and Theorem \ref{T-4.3} it follows that
\begin{equation}\label{F-4.8}
\begin{aligned}
D_\alpha(\mu_\rho\|\mu_\sigma)&\le D_\alpha^*(\tilde\rho\|\tilde\sigma)=D_\alpha^*(\rho\|\sigma),
\qquad\alpha\in[1/2,+\infty)\setminus\{1\}, \\
D_\alpha(\mu_\rho\|\mu_\sigma)&\le D_\alpha(\tilde\rho\|\tilde\sigma)=D_\alpha(\rho\|\sigma),
\qquad\alpha\in[0,2]\setminus\{1\}.
\end{aligned}
\end{equation}
Therefore, when $\cA$ is $G$-central and $\rho,\sigma\in\cS_G(\cA)$, it follows from \eqref{F-4.6}
and \eqref{F-4.8} that $D_\alpha^*(\rho\|\sigma)=D_\alpha(\mu_\rho\|\mu_\sigma)$ for all
$\alpha\in[1/2,+\infty)\setminus\{1\}$ and $D_\alpha(\rho\|\sigma)=D_\alpha(\mu_\rho\|\mu_\sigma)$
for all $\alpha\in[0,2]\setminus\{1\}$, showing also $D(\rho\|\sigma)=D(\mu_\rho\|\mu_\sigma)$
as shown in \cite[Theorem 3.2]{HOT}.
\end{example}

\begin{example}\label{E-4.8}\rm
Let $\B=\bigotimes_1^\infty\bM_d$ be a UHF (or one-dimensional spin) $C^*$-algebra, and
$\tau$ be the action of $S_\infty$ (the group of finite permutations on $\bN$) given by
$\tau\bigl(\bigotimes_1^\infty a_n\bigr)=\bigotimes_1^\infty a_{g(n)}$, $g\in S_\infty$. Let $K$ be
a compact group and $u_k$ ($k\in K$) a continuous unitary representation of $K$ on $\bC^d$,
so a product action $\beta$ of $K$ on $\B$ is defined by
$\beta_k:=\bigotimes_1^\infty\mathrm{Ad}(u_k)$ (where $\mathrm{Ad}(u_k)=u_k(\cdot)u_k^*$).
The $\beta$-fixed point $C^*$-subalgebra $\cA:=\B^\beta$ of $\B$ is called a
\emph{gauge-invariant $C^*$-algebra}. Let $\B_n:=\bigotimes_1^n\bM_d$ and
$\cA_n:=\B_n^\beta=\cA\cap\B_n$ for $n\ge1$; then $\bigcup_{n=1}^\infty\cA_n$ is norm-dense
in $\cA$ (so $\cA$ is an AF $C^*$-algebra.) Then $(\B,S_\infty,\tau)$ and $(\cA,S_\infty,\tau|_\cA)$
are typical cases of Example \ref{E-4.7}, since $\B$ and $\cA$ are asymptotically abelian with
respect to $S_\infty$. St\o rmer's theorem \cite{Stor} says that the extreme points of
$\cS_{S_\infty}(\B)$ are the symmetric product states $\bigotimes_1^\infty\psi$,
$\psi\in\cS(\bM_d)$. For every $\rho,\sigma\in\cS_{S_\infty}(\cA)$, by Proposition \ref{P-4.5}\,(4)
and Example \ref{E-4.7} one has
\begin{equation}\label{F-4.9}
\begin{aligned}
D_\alpha^*(\rho\|\sigma)&=\lim_{n\to\infty}D_\alpha^*(\rho|_{\cA_n}\|\sigma|_{\cA_n})
=D_\alpha(\mu_\rho\|\mu_\sigma),\qquad\alpha\in[1/2,+\infty)\setminus\{1\}, \\
D_\alpha(\rho\|\sigma)&=\lim_{n\to\infty}D_\alpha(\rho|_{\cA_n}\|\sigma|_{\cA_n})
=D_\alpha(\mu_\rho\|\mu_\sigma),\qquad\alpha\in[0,2]\setminus\{1\},
\end{aligned}
\end{equation}
where $\mu_\rho,\mu_\sigma$ are the representing measures of $\rho,\sigma$ on the extreme
boundary of $\cS_{S_\infty}(\cA)$. In particular, it follows that
$D_\alpha^*(\rho\|\sigma)=D_\alpha(\rho\|\sigma)$ for any $\rho,\sigma\in\cS_{S_\infty}(\cA)$
and all $\alpha\in[1/2,2]\setminus\{1\}$ even though
$D_\alpha^*(\rho|_{\cA_n}\|\sigma|_{\cA_n})<D_\alpha(\rho|_{\cA_n}\|\sigma|_{\cA_n})$ unless
$\rho|_{\cA_n}$ and $\sigma|_{\cA_n}$ commute. This phenomenon can naturally be understood
by considering the asymptotic abeliannes of $\cA$ with respect to $S_\infty$. Now let $\cT(\cA)$
denote the set of tracial states on $\cA$. Recall \cite[Theorem 3.2]{Pr} that any extreme point of
$\cT(\cA)$ is the restriction to $\cA$ of a symmetric product state of $\B$, so that
$\cT(\cA)\subseteq\cS_{S_\infty}(\cA)$. Furthermore, note (see \cite[Sec.~4]{Ta-Wi}) that
$\cT(\cA)$ is a face of $\cS_{S_\infty}(\cA)$. For every $\rho,\sigma\in\cT(\cA)$,
since tracial states $\rho,\sigma$ commute, one has
$D_\alpha^*(\rho\|\sigma)=D_\alpha(\rho\|\sigma)$ for all $\alpha\in[1/2,+\infty)\setminus\{1\}$ by
\cite[Remark 3.18\,(2)]{Hi}, and hence by \eqref{F-4.9},
\begin{align}\label{F-4.10}
D_\alpha(\rho\|\sigma)=\lim_{n\to\infty}D_\alpha(\rho|_{\cA_n}\|\sigma|_{\cA_n})
=D_\alpha(\mu_\rho\|\mu_\sigma),\qquad\alpha\in[0,+\infty)\setminus\{1\},
\end{align}
where the representing measures $\mu_\rho,\mu_\sigma$ are supported on the extreme
boundary of $\cT(\cA)$ . For example, let $d=2$,
$K=\bT=\{\zeta\in\bC:|\zeta|=1\}$ and $u_\zeta:=\begin{bmatrix}1&0\\0&\zeta\end{bmatrix}$,
$\zeta\in\bT$. Then $\cA=\B^\beta$ is the so-called \emph{GICAR algebra}, and note that
\[
\cA_n=\bigoplus_{k=0}^n\bM_{n\choose k},\qquad n\ge1.
\]
It is well known that the extreme boundary of $\cT(\cA)$ is parametrized by $\lambda\in[0,1]$;
to be precise, it is $\{\omega_\lambda:\lambda\in[0,1]\}$, where
$\omega_\lambda=\bigl(\bigotimes_1^\infty\psi_\lambda\bigr)|_{\cA}$ with
$\psi_\lambda:=\Tr\left(\begin{bmatrix}\lambda&0\\0&1-\lambda\end{bmatrix}\cdot\right)$ on
$\bM_2$. Let $\rho,\sigma\in\cT(\cA)$ and decompose them as
$\rho=\int_0^1\omega_\lambda\,d\mu_\rho(\lambda)$,
$\sigma=\int_0^1\omega_\lambda\,d\mu_\sigma(\lambda)$ with unique probability measures
$\mu_\rho,\mu_\sigma$ on $[0,1]$. Since
\[
\omega_\lambda|_{\cA_n}=\bigoplus_{k=0}^n\lambda^k(1-\lambda)^{n-k}\Tr_{n\choose k}
\]
(with convention $\lambda^0=1$ for $\lambda=0$), the convergence
$\lim_{n\to\infty}Q_\alpha(\rho|_{\cA_n}\|\sigma|_{\cA_n})=Q_\alpha(\mu_\rho\|\mu_\sigma)$
due to \eqref{F-4.10} can be rewritten in an explicit form as
\begin{align*}
&\lim_{n\to\infty}\sum_{k=0}^n{n\choose k}
\biggl(\int_0^1\lambda^k(1-\lambda)^{n-k}\,d\mu_1(\lambda)\biggr)^\alpha
\biggl(\int_0^1\lambda^k(1-\lambda)^{n-k}\,d\mu_2(\lambda)\biggr)^{1-\alpha} \nonumber\\
&\qquad
=\int_0^1\biggl({d\mu_1\over d(\mu_1+\mu_2)}\biggr)^\alpha
\biggl({d\mu_2\over d(\mu_1+\mu_2)}\biggr)^{1-\alpha}\,d(\mu_1+\mu_2),
\qquad\alpha\in[0,+\infty),
\end{align*}
for any probability measures $\mu_1,\mu_2$ on $[0,1]$. (It does not seem easy to verify this
convergence formula in a direct manner.)
\end{example}

\subsection{The strong converse exponent in nuclear $C^*$-algebras}\label{Sec-4.2}

Let $\cA$ and $\rho,\sigma\in\cA_+^*$ be as in Section \ref{Sec-4.1}. For each $r\in\bR$,
we define the \emph{Hoeffding anti-divergence} $H_r^*(\rho\|\sigma)$ in the same way as in
Definition \ref{D-3.6} with use of $D_\alpha^*(\rho\|\sigma)$, i.e.,
$\widehat D_\alpha^*(\rho\|\sigma)$ in \eqref{F-4.1} or \eqref{F-4.4}. For any
$(\rho,\sigma)$-normal representation $\pi$ of $\cA$, Theorem \ref{T-4.3} gives
\begin{align}\label{F-4.11}
H_r^*(\rho\|\sigma)=H_r^*(\rho_\pi\|\sigma_\pi),\qquad r\in\bR.
\end{align}

We furthermore define strong converse exponents as in Definition \ref{D-3.5} in the $C^*$-algebra
setting. Here, recall (see, e.g., \cite[Sec.~IV.4]{Ta1}) that for $C^*$-algebras $\cA$ and $\B$, a
$C^*$-cross-norm on the algebraic tensor product $\cA\odot\B$ is not unique in general (though
it exists always), and we have the smallest one $\norm{\cdot}_{\min}$ and the largest one
$\norm{\cdot}_{\max}$. The \emph{minimal (or spatial) $C^*$-tensor product}
$\cA\otimes_{\min}\B$ is the completion of $\cA\odot\B$ with respect to $\norm{\cdot}_{\min}$,
and the \emph{maximal $C^*$-tensor product} $\cA\otimes_{\max}\B$ is that with respect to
$\norm{\cdot}_{\max}$. To discuss the simple hypothesis testing on a $C^*$-algebra, the minimal
$C^*$-tensor product is suitable for the following reasons:
When $\cA$ and $\B$ are realized in $\B(\cH_1)$ and $\B(\cH_2)$ respectively,
$\cA\otimes_{\min}\B$ is isomorphic to the $C^*$-algebra generated by $\cA\odot\B$ in
$\B(\cH_1\otimes\cH_2)$; see \cite[Theorem IV.4.9\,(iii)]{Ta1}. For any representations
$\pi_1$ of $\cA$ and $\pi_2$ of $\B$, the tensor product representation $\pi_1\otimes\pi_2$ of
$\cA\otimes_{\min}\B$ satisfies
$(\pi_1\otimes\pi_2)(\cA\otimes_{\min}\B)''=\pi_1(\cA)''\,\overline\otimes\,\pi_2(\B)''$;
see \cite[Proposition IV.4.13]{Ta1}.

Let $\rho,\sigma\in\cA_+^*$. For each $n\in\bN$ let $\cA^{\otimes n}$ ($=\cA^{\otimes_{\min}n}$)
be the $n$-fold minimal $C^*$-tensor product of $\cA$, and $\rho_n:=\rho^{\otimes n}$ (resp.,
$\sigma_n:=\sigma^{\otimes n}$) be the $n$-fold tensor product of $\rho$ (resp., $\sigma$) on
$\cA^{\otimes n}$. The following are the strong converse exponents in the setting of the simple
hypothesis testing for the null hypothesis $H_0:\rho$ versus the alternative hypothesis
$H_1:\sigma$.

\begin{definition}\label{D-4.9}\rm
For each $r\in\bR$ we define the \emph{strong converse exponents}
$\underline{sc}_r(\rho\|\sigma)$ etc.\ in the same expressions as in Definition \ref{D-3.5}, where
$\{T_n\}$ in the present setting are taken as $T_n\in\cA^{\otimes n}$, $0\le T_n\le1$, $n\in\bN$.
Then the relation \eqref{F-3.5} is obvious as before.
\end{definition}

\begin{remark}\label{R-4.10}\rm
The most general notion of a test in the $C^*$-algebra setting is not covered by tests taken in
Definition \ref{D-4.9}. For instance, one could say that a measurement is represented by a self-adjoint
element $X$ in $\cA$, and the possible outcomes are the points of the spectrum of $X$. After
measuring $X$, one could always make a classical post-processing, e.g., by dividing the spectrum
of $X$ into two disjoint sets $B_0$ and $B_1$. If the measurement outcome falls into $B_0$, we
accept $H_0$, otherwise we accept $H_1$. The error probabilities corresponding to such a test may
not be described by an element $T\in\cA$ with $0\le T\le I$. However, the next lemma shows that
considering such more general tests does not change the strong converse exponents.
\end{remark}

\begin{lemma}\label{L-4.11}
Let $\rho,\sigma\in\cA_+^*$ and $\pi$ be a $(\rho,\sigma)$-normal representation of $\cA$. For
every $r\in\bR$, let $\eps_r(\rho\|\sigma)$ be any strong converse exponent given in Definition
\ref{D-4.9}, and $\eps_r(\rho_{\pi}\|\sigma_{\pi})$ be the same strong converse exponent in
Definition \ref{D-3.5} for $\rho_{\pi},\sigma_{\pi}$. Then we have
\begin{align}\label{F-4.12}
\eps_r(\rho\|\sigma)=\eps_r(\rho_{\pi}\|\sigma_{\pi}).
\end{align}
\end{lemma}

\begin{proof}
Let $\M:=\pi(\cA)''$. For every $T_n\in\cA^{\otimes n}$, $n\in\bN$, with $0\le T_n\le1$, one has
$\pi^{\otimes n}(T_n)\in\M^{\overline\otimes n}$, $0\le\pi^{\otimes n}(T_n)\le 1$,
$\rho_{\pi,n}(\pi^{\otimes n}(T_n))=\rho_n(T_n)$ and
$\sigma_{\pi,n}(\pi^{\otimes n}(T_n))=\sigma_n(T_n)$, where
$\rho_{\pi,n}:=(\rho_\pi)^{\otimes n}$ and $\sigma_{\pi,n}:=(\sigma_\pi)^{\otimes n}$.
Hence $\eps_r(\rho_\pi\|\sigma_\pi)\le\eps_r(\rho\|\sigma)$ holds immediately.
Conversely, for every $\widetilde T_n\in\M^{\overline\otimes n}=\pi^{\otimes n}(\cA^{\otimes n})''$,
$n\in\bN$, with $0\le\widetilde T_n\le1$, by Kaplansky's density theorem and
\cite[Lemma IV.3.8]{Is} one can choose $T_n\in\cA^{\otimes n}$, $n\in\bN$, such that
$0\le T_n\le1$ and
\begin{align*}
\rho_n(T_n)=\rho_{\pi,n}(\pi^{\otimes n}(T_n))
&\in\begin{cases}
(e^{-1}\rho_{\pi,n}(\widetilde T_n),e\rho_{\pi,n}(\widetilde T_n)) &
\text{if $\rho_{\pi,n}(\widetilde T_n)>0$}, \\
[0,1/n^n) & \text{if $\rho_{\pi,n}(\widetilde T_n)=0$},
\end{cases} \\
\sigma_n(T_n)=\sigma_{\pi,n}(\pi^{\otimes n}(T_n))
&\in\begin{cases}
(e^{-1}\sigma_{\pi,n}(\widetilde T_n),e\sigma_{\pi,n}(\widetilde T_n)) &
\text{if $\sigma_{\pi,n}(\widetilde T_n)>0$}, \\
[0,1/n^n) & \text{if $\sigma_{\pi,n}(\widetilde T_n)=0$},
\end{cases}
\end{align*}
which imply that
\begin{align*}
-{1\over n}\log\rho_n(T_n)
&\in\begin{cases}
\bigl(-{1\over n}-{1\over n}\log\rho_{\pi,n}(\widetilde T_n),
{1\over n}-{1\over n}\log\rho_{\pi,n}(\widetilde T_n)\bigr) &
\text{if $\rho_{\pi,n}(\widetilde T_n)>0$}, \\
(\log n,+\infty] & \text{if $\rho_{\pi,n}(\widetilde T_n)=0$},
\end{cases} \\
-{1\over n}\log\sigma_n(T_n)
&\in\begin{cases}
\bigl(-{1\over n}-{1\over n}\log\sigma_{\pi,n}(\widetilde T_n),
{1\over n}-{1\over n}\log\sigma_{\pi,n}(\widetilde T_n)\bigr) &
\text{if $\sigma_{\pi,n}(\widetilde T_n)>0$}, \\
(\log n,+\infty] & \text{if $\sigma_{\pi,n}(\widetilde T_n)=0$}.
\end{cases}
\end{align*}
From these one can see that $\eps_r(\rho\|\sigma)\le\eps_r(\tilde\rho\|\tilde\sigma)$. Hence
\eqref{F-4.12} follows.
\end{proof}

The next theorem is the $C^*$-algebra version of Theorem \ref{T-3.7} under the assumption
of the generated von Neumann algebra being injective, in particular, when $\cA$ is nuclear. For
the convenience of the reader, we recall the notion of nuclear $C^*$-algebras in Appendix
\ref{Sec-E}. Here we restrict $\rho,\sigma$ to states on $\cA$ for the same reason as before.

\begin{theorem}\label{T-4.12}
Let $\rho,\sigma\in\cA_+^*$ be states such that $D_\alpha^*(\rho\|\sigma)<+\infty$ for some
$\alpha>1$. If there exists a $(\rho,\sigma)$-normal representation $\pi$  of $\cA$ such that
$\pi(\cA)''$ is injective, then all the equalities in \eqref{F-3.8} hold for every $r\in\bR$ in
the present $C^*$-algebra situation too. In particular, this is the case if $\cA$ is nuclear.
\end{theorem}

\begin{proof}
Let $\pi$ be as stated in the theorem. Then Theorem \ref{T-3.7} says that the equalities in
\eqref{F-3.8} for $\rho_\pi,\rho_\sigma$ in place of $\rho,\sigma$ hold for every $r\in\bR$.
Hence the assertion follows from Lemma \ref{L-4.11} and \eqref{F-4.11}.
\end{proof}

We remark that the same formula as in Theorem \ref{T-3.9} holds true under the assumption
of Theorem \ref{T-4.12}, where $\alpha_{e^{-nr}}^*(\rho_n\|\sigma_n)$ is defined as in
\eqref{F-3.14} with $T_n\in\cA^{\otimes n}$, $0\le T_n\le1$. The proof is the same as that of
Theorem \ref{T-3.9}. Also, the generalized $\kappa$-cutoff rate $C_\kappa(\rho\|\sigma)$ in
Definition \ref{D-3.10} makes sense for states $\rho,\sigma\in\cA_+^*$, and Theorem \ref{T-3.11}
is extended to the $C^*$-algebra setting with the same proof.

In the rest of the section let us discuss the (regularized) measured and the (regularized)
test-measured R\'enyi divergences in the $C^*$-algebra setting. A measurement in $\cA$ is
given by $\fM=(M_j)_{1\le j\le k}$ of $M_j\in\cA_+$ such that $\sum_{j=1}^kM_j=1$. Then
$D_\alpha^\meas(\rho\|\sigma)$ and $\overline D_\alpha^\meas(\rho\|\sigma)$ for
$\rho,\sigma\in\cA_+^*$ are defined in the same way as in Definition \ref{D-3.12} by taking
measurements in $\cA$ instead of $\M$. The test-measured versions
$D_\alpha^\test(\rho\|\sigma)$ and $\overline D_\alpha^\test(\rho\|\sigma)$ for
$\rho,\sigma\in\cA_+^*$ are also defined in the same way as just after Definition \ref{D-3.12} with
tests $T\in\cA$, $0\le T\le1$.

\begin{lemma}\label{L-4.13}
Let $\rho,\sigma\in\cA_+^*$ and $\pi$ be any $(\rho,\sigma)$-normal representation of $\cA$.
Then for every $\alpha\in[1/2,+\infty)\setminus\{1\}$ we have
\begin{align}
D_\alpha^\meas(\rho\|\sigma)&=D_\alpha^\meas(\rho_\pi\|\sigma_\pi),\qquad
\overline D_\alpha^\meas(\rho\|\sigma)=\overline D_\alpha^\meas(\rho_\pi\|\sigma_\pi),
\label{F-4.13}\\
D_\alpha^\test(\rho\|\sigma)&=D_\alpha^\test(\rho_\pi\|\sigma_\pi),\ \ \,\qquad
\overline D_\alpha^\test(\rho\|\sigma)=\overline D_\alpha^\test(\rho_\pi\|\sigma_\pi),
\label{F-4.14}
\end{align}
\end{lemma}

\begin{proof}
Let $\M:=\pi(\cA)''$. If $(M_i)_{1\le i\le k}$ is a measurement in $\cA$, then
$(\pi(M_i))_{1\le i\le k}$ is a measurement in $\M$. Hence
$D_\alpha^\meas(\rho\|\sigma)\le D_\alpha^\meas(\rho_\pi\|\sigma_\pi)$ holds immediately.
Conversely, let $\widetilde\fM=(\widetilde M_i)_{1\le i\le k}$ be a measurement in $\M$.
Consider the representation $\pi_k=\pi\otimes\mathrm{id}_k$ of $\bM_k(\cA)=\cA\otimes\bM_k$
on $\cH\otimes\bC^k$ (where $\bM_k$ is the $k\times k$ matrix algebra and $\cH$ is the
representing Hilbert space of $\pi$). Noting that $\pi_k(\bM_k(\cA))''=\bM_k(\pi(\cA))''=\bM_k(\M)$,
we define $\tilde A\in\bM_k(\M)$ by
\[
\tilde A:=\begin{bmatrix}
\widetilde M_1^{1/2}&\\
\vdots&\ \ \mbox{\huge{0}}\quad\\
\widetilde M_k^{1/2}&
\end{bmatrix}.
\]
Note that $\tilde A^*\tilde A$ is a projection with the $(1,1)$-block $1$ and all other blocks
$0$. Hence $\tilde A$ is a contraction. By Kaplansky's density theorem and \cite[Lemma IV.3.8]{Is}
there exists a net $\{A^{(\lambda)}\}$ of contractions in $\bM_k(\cA)$ such that
$\pi_k\bigl(A^{(\lambda)}\bigr)\to\tilde A$ in the strong* topology. Write
$A^{(\lambda)}=\bigl[a_{ij}^{(\lambda)}\bigr]_{i,j=1}^k$ with $a_{ij}^{(\lambda)}\in\cA$; then
$\pi_k\bigl(A^{(\lambda)}\bigr)=\bigl[\pi(a_{ij}^{(\lambda)})\bigr]_{i,j=1}^k$ so that
$\pi(a_{i1}^{(\lambda)})\to\widetilde M_i^{1/2}$ in the strong* topology for $1\le i\le k$. Set
$M_i^{(\lambda)}:=\bigl(a_{i1}^{(\lambda)}\bigr)^*a_{i1}^{(\lambda)}\in\cA_+$, $1\le i\le k$. Then
$\pi\bigl(M_i^{(\lambda)}\bigr)\to\widetilde M_i$ strongly for $1\le i\le k$, and
$\sum_{i=1}^kM_i^{(\lambda)}=\sum_{i=1}^k\bigl(a_{i1}^{(\lambda)}\bigr)^*a_{i1}^{(\lambda)}$ is
the $(1,1)$-block of $\bigl(A^{(\lambda)}\bigr)^*A^{(\lambda)}$ so that
$\sum_{i=1}^kM_i^{(\lambda)}\le1$. Therefore, one can define a net $\{\fM_\lambda\}$ of
measurements in $\cA$ by $\fM_\lambda=\bigl(M_i^{(\lambda)}\bigr)_{1\le i\le k+1}$, where
$M_{k+1}^{(\lambda)}:=1-\sum_{i=1}^kM_i^{(\lambda)}$. Since
$\pi\bigl(M_{k+1}^{(\lambda)}\bigr)\to 1-\sum_{i=1}^k\widetilde M_i=0$, we have
\[
\fM_\lambda(\rho)=\bigl(\rho_\pi\circ\pi\bigl(M_i^{(\lambda)}\bigr)\bigr)_{1\le i\le k+1}
\,\longrightarrow\,\widetilde\fM(\rho_\pi)\oplus0,
\]
and similarly $\fM_\lambda(\sigma)\to\widetilde\fM(\sigma_\pi)\oplus0$. From the lower
semi-continuity of (classical) R\'enyi divergence, it follows that
\[
D_\alpha(\widetilde\fM(\rho_\pi)\|\widetilde\fM(\sigma_\pi))
\le\liminf_\lambda D_\alpha(\fM_\lambda(\rho)\|\fM_\lambda(\sigma))
\le D_\alpha^\meas(\rho\|\sigma).
\]
Hence $D_\alpha^\meas(\rho_\pi\|\sigma_\pi)\le D_\alpha^\meas(\rho\|\sigma)$, implying the
first equality of \eqref{F-4.13}. Moreover, applying this to $\rho^{\otimes n},\sigma^{\otimes n}$
and the representation $\pi^{\otimes n}$ of $\cA^{\otimes n}$, we have for every $n\in\bN$,
\[
D_\alpha^\meas(\rho^{\otimes n}\|\sigma^{\otimes n})
=D_\alpha^\meas(\rho_\pi^{\otimes n}\|\sigma_\pi^{\otimes n}),
\]
since $\rho_\pi^{\otimes n},\sigma_\pi^{\otimes n}$ are the respective normal extensions of
$\rho^{\otimes n},\sigma^{\otimes n}$ to
$\pi^{\otimes n}(\cA^{\otimes n})''=\M^{\overline\otimes n}$. Thus, the second equality of
\eqref{F-4.13} holds as well. The proof of \eqref{F-4.14} is similar (and easier) by approximating
tests in $\M$ with those in $\cA$ (similarly to the proof of Lemma \ref{L-4.11}).
\end{proof}

\begin{proposition}\label{P-4.14}
Assume that there exists a $(\rho,\sigma)$-normal representation $\pi$ of $\cA$ such that
$\pi(\cA)''$ is injective (in particular, this is the case if $\cA$ is nuclear). 
Then \eqref{F-3.17} holds for every $\alpha\in[1/2,+\infty)\setminus\{1\}$ and \eqref{F-3.18}
holds for every $\alpha>1$.
\end{proposition}
\begin{proof}
By Theorem \ref{T-4.3} and Lemma \ref{L-4.13}, one can apply Proposition \ref{P-3.13} to 
$\rho_\pi,\sigma_\pi$ to obtain the assertion. 
\end{proof}

\section*{Acknowledgments}
The work of M.M.\ was partially funded by the National Research, Development and
Innovation Office of Hungary via the research grants K124152 and KH129601, and by the
Ministry of Innovation and Technology and the National Research, Development and Innovation
Office within the Quantum Information National Laboratory of Hungary.

\appendix

\section{Relative modular operators}\label{Sec-A}
\label{sec:relmodop}

Let $\M$ be a general von Neumann algebra with the predual $\M_*$, and $\M_*^+$ be the
positive part of $\M_*$ consisting of normal positive linear functionals on $\M$. We consider
$\M$ in its \emph{standard form} $(\M,\cH,J,\cP)$ \cite{Ha}, that is, $\M$ is represented on a
Hilbert space $\cH$ with the \emph{modular conjugation} (a conjugate-linear involution) $J$ and
the \emph{natural cone} (a self-dual cone) $\cP$, satisfying the following properties:
\begin{itemize}
\item[(1)] $JMJ=M'$ ($M'$ being the commutant of $\M$),
\item[(2)] $JxJ=x^*$, $x\in\M\cap\M'$ (the center of $\M$),
\item[(3)] $J\xi=\xi$, $\xi\in\cP$,
\item[(4)] $xJxJ\cP\subseteq\cP$, $x\in\M$.
\end{itemize}
Any von Neumann algebra has a unique (up to unitary conjugation) standard form; see
\cite[Theorem 2.3]{Ha}. Any $\sigma\in\M_*^+$ has a unique \emph{vector representative}
$\xi_\sigma$ in $\cP$ so that $\sigma(x)=\<\xi_\sigma,x\xi_\sigma\>$, $x\in\M$. The support
$s(\sigma)=s_\M(\sigma)\in\M$ of $\sigma$ is the orthogonal projection onto
$\overline{\M'\xi_\sigma}$, while the $\M'$-support $s_{\M'}(\sigma)\in\M'$ is the orthogonal
projection onto $\overline{\M\xi_\sigma}$ so that $s_{\M'}(\sigma)=Js_\M(\sigma)J$.

For any $\rho,\sigma\in\M_*^+$, the closable conjugate-linear operators $S_{\rho,\sigma}$ and
$F_{\rho,\sigma}$ are defined by
\begin{align*}
S_{\rho,\sigma}(x\xi_\sigma+\eta)&:=s_\M(\sigma)x^*\xi_\sigma,\ \ \qquad
x\in\M,\ \eta\in(1-s_{\M'}(\sigma))\cH, \\
F_{\rho,\sigma}(x'\xi_\sigma+\zeta)&:=s_{\M'}(\sigma)x'^*\xi_\sigma,\qquad
x'\in\M',\ \zeta\in(1-s_\M(\sigma))\cH,
\end{align*}
for which $S_{\rho,\sigma}^*=\overline F_{\rho,\sigma}$. The \emph{relative modular operator}
$\Delta_{\rho,\sigma}$ \cite{Ar2} is
\[
\Delta_{\rho,\sigma}:=S_{\rho,\sigma}^*\overline S_{\rho,\sigma}
\]
and the polar decomposition of $\overline S_{\rho,\sigma}$ is given as
$\overline S_{\rho,\sigma}=J\Delta_{\rho,\sigma}^{1/2}$. When $\rho=\sigma$,
$\Delta_{\sigma,\sigma}$ is the \emph{modular operator} $\Delta_\sigma$.

When $\M=B(\cH)$ on a Hilbert space $\cH$, consider the Hilbert--Schmidt class $\cC_2(\cH)$
with the Hilbert--Schmidt inner product $\<X,Y\>_\HS:=\Tr(X^*Y)$, $X,Y\in\cC_2(\cH)$; then
the standard form of $B(\cH)$ is given as
\[
(B(\cH),\cC_2(\cH), J=\,^*,\cC_2(\cH)_+),
\]
where $B(\cH)$ is represented on $\cC_2(\cH)$ by the left multiplications $L_AX:=AX$ for
$A\in B(\cH)$, $X\in\cC_2(\cH)$, and $\cC_2(\cH)_+:=\{X\in\cC_2(\cH):X\ge0\}$. Each
$\rho\in B(\cH)_*^+$ is identified with a trace-class operator $\hat\rho\ge0$ so that
$\rho(X)=\Tr(\hat\rho X)=\<\hat\rho^{1/2},X\hat\rho^{1/2}\>_\HS$, $X\in B(\cH)$, and
$\hat\rho^{1/2}\in\cC_2(\cH)_+$ is the vector representative of $\rho$. 
For $\rho,\sigma\in B(\cH)_*^+$ the relative modular operator $\Delta_{\rho,\sigma}$ is written
as $\Delta_{\rho,\sigma}=L_{\hat\rho}R_{\hat\sigma^{-1}}$, where $\hat\sigma^{-1}$ is the
generalized inverse (i.e., the inverse with restriction to the support $s(\sigma)\cH$) of
$\hat\sigma$ and $R_{\hat\sigma^{-1}}$ is the right multiplication by $\hat\sigma^{-1}$. 
Of course, when $\dim\cH<+\infty$, we have $\cC_2(\cH)=B(\cH)$.

\section{Haagerup's $L^p$-spaces}\label{Sec-B}

Assume that $\M$ is $\sigma$-finite, i.e., there exists a faithful $\omega\in\M_*^+$.
Let us denote by$\N$ the crossed product $\M\rtimes_\omega\bR$ of $\M$ by the
\emph{modular automorphism group}
$\sigma_t^\omega=\Delta_\omega^{it}(\cdot)\Delta_\omega^{-it}$, $t\in\bR$. Le $\theta_s$,
$s\in\bR$, be the \emph{dual action} of $\N$ so that
$\tau\circ\theta_s=e^{-s}\tau$, $s\in\bR$, where $\tau$ is the \emph{canonical trace} on $\N$;
the crossed product construction was developed in the structure theory of von Neumann algebras
\cite{Ta4}. Let $\widetilde\N$ denote the space of \emph{$\tau$-measurable operators} \cite{Ne,Te}
affiliated with $\N$. For each $p\in(0,+\infty]$, \emph{Haagerup's $L^p$-space} $L^p(\M)$
\cite{Te} is defined by
\[
L^p(\M):=\{x\in\widetilde\N:\theta_s(x)=e^{-s/p}x,\,s\in\bR\}
\]
(in particular, $L^\infty(\M)=\M$), whose positive part is
$L^p(\M)_+:=L^p(\M)\cap\widetilde\N_+$. There exists an order isomorphism
$\M_*\cong L^1(\M)$, given as $\psi\in\M_*\mapsto h_\psi\in L^1(\M)$, so that
$\tr(h_\psi):=\psi(\1)$, $\psi\in\M_*$, defines a positive linear functional $\tr$ on $L^1(\M)$. For
$1\le p<+\infty$ the $L^p$-norm $\|a\|_p$ of $a\in L^p(\M)$ is given by $\|a\|_p:=\tr(|a|^p)^{1/p}$,
and the $L^\infty$-norm $\norm{\cdot}_\infty$ is the operator norm on $\M$. For $1\le p<+\infty$,
$L^p(\M)$ is a Banach space with the norm $\norm{\cdot}_p$, whose dual Banach space is
$L^q(\M)$, where $1/p+1/q=1$, by the duality
\[
(a,b)\in L^p(\M)\times L^q(\M)\,\longmapsto\,\tr(ab)\ (=\tr(ba)).
\]
In particular, $L^2(\M)$ is a Hilbert space with the inner product $\<a,b\>=\tr(a^*b)$ ($=\tr(ba^*)$).
Then
\[
(\M,L^2(\M),J=\,^*,L^2(\M)_+)
\]
becomes a \emph{standard form} of $\M$, where $\M$ is represented on $L^2(\M)$ by the left
multiplication. Each $\rho\in\M_*^+$ is represented as
\[
\rho(x)=\tr(xh_\rho)=\<h_\rho^{1/2},xh_\rho^{1/2}\>,\qquad x\in\M,
\]
with the vector representative $h_\rho^{1/2}\in L^2(\M)_+$. Note that the support projection
$s(\rho)$ ($\in\M$) of the functional $\rho$ coincides with that of the operator $h_\rho$. For any
projection $e\in\M$, Haagerup's $L^p$-space $L^p(e\M e)$ is identified with $eL^p(\M)e$ and
the standard form of $e\M e$ is given by $(e\M e,eL^2(\M)e,J=\,^*,eL^2(\M)_+e)$.

Note that $L^p(\M)$ is independent (up to isometric isomorphism) of the choice of $\omega$
(where $\omega$ can be a faithful normal semifinite \emph{weight} unless $\M$ is
$\sigma$-finite), and that when $\M$ is semifinite with a faithful normal semifinite trace $\tau_0$,
$L^p(\M)$ can be identified with the tracial $L^p$-space $L^p(\M,\tau_0)$ (see, e.g., \cite{Ne}).
In particular, when $\M=B(\cH)$ with $\omega=\Tr$ (and so $\Delta_\omega=\1$), note that
$\N=\M\overline\otimes L^\infty(\bR)$ on $\cH\otimes L^2(\bR)$ and the canonical trace on
$\N$ is $\tau=\Tr\otimes\int_\bR(\cdot)e^t\,dt$, so that $L^p(\M)=\cC_p(\cH)\otimes e^{-t/p}$
with $\|X\otimes e^{-t/p}\|_{L^p(\M)}=\|X\|_p$ for $X\in\cC_p(\cH)$. Here, the symbol $e^{-t/p}$ is
used to denote the multiplication operator on $L^2(\bR)$, and $\cC_p(\cH)$ is the
Schatten--von Neumann $p$-class with $\|X\|_p:=(\Tr\,|X|^p)^{1/p}$. Therefore, $L^p(\M)$
coincides with $\cC_p(\cH)$ by just neglecting the superfluous tensor factor $e^{-t/p}$; see
\cite[Remark 8.16, Example 9.11]{Hi6} for more details on this matter.

It might be instructive to note that Haagerup's $L^p(\M)$ is different from the tracial
$L^p$-space $L^p(\N,\tau)$ with the canonical trace $\tau$, even when $\M=B(\cH)$. In this case,
$L^p(\M)=\cC_p(\cH)\otimes e^{-t/p}$ as stated above, and for every $X\in\cC_p(\cH)$,
\[
\|X\otimes e^{-t/p}\|_{L^p(\N,\tau)}=\|X\|_p\Bigl(\int_\bR(e^{-t/p})^pe^t\,dt\Bigr)^{1/p}
=\|X\|_p\Bigl(\int_\bR dt\Bigr)^{1/p}=+\infty
\]
unless $X=0$. However, in the general case of $\M$, the exact relation of elements in $L^p(\M)$
with the canonical trace $\tau$ on $\N$ is expressed as follows: for every $a\in L^p(\M)$ and
$p\in(0,+\infty)$,
\[
\mu_t(a)=t^{-1/p}\|a\|_p,\qquad t>0,
\]
where $\mu_t(a)$ is the $t$th \emph{generalized $s$-number} of $a$ with respect to $\tau$; see
\cite[Lemma 4.8]{FK} and \cite[Lemma 9.14]{Hi6}. The above expression is sometimes useful
though it is not used in this paper.

\section{Kosaki's interpolation $L^p$-spaces}\label{Sec-C}

Assume that $\M$ is $\sigma$-finite and let a faithful $\omega\in\M_*^+$ be given with
$h_{\omega}\in L^1(\M)_+$. Consider an embedding $\M$ into $L^1(\M)$ by
$x\mapsto h_\omega^{1/2}xh_\omega^{1/2}$. Defining
$\|h_\omega^{1/2}xh_\omega^{1/2}\|_\infty:=\|x\|_\infty$ on $h_\omega^{1/2}\M h_\omega^{1/2}$
we have a pair $(h_\omega^{1/2}\M h_\omega^{1/2},L^1(\M))$ of compatible Banach spaces
(see, e.g., \cite{BL}). For $1<p<+\infty$ \emph{Kosaki's (symmetric) $L^p$-space}
$L^p(\M,\omega)$ \cite{Ko1} with respect to $\omega$ is the complex interpolation Banach space
\[
C_{1/p}(h_\omega^{1/2}\M h_\omega^{1/2},L^1(\M))
\]
equipped with the interpolation norm $\norm{\cdot}_{p,\omega}$ ($=\norm{\cdot}_{C_{1/p}}$)
\cite{BL}. Moreover, $L^1(\M,\omega):=L^1(\M)$ with $\norm{\cdot}_{1,\omega}=\norm{\cdot}_1$
and $L^\infty(\M,\omega):=h_\omega^{1/2}\M h_\omega^{1/2}$ ($\cong\M$) with
$\norm{\cdot}_{\infty,\omega}=\norm{\cdot}_\infty$. Kosaki's theorem \cite[Theorem 9.1]{Ko1} says
that for every $p\in[1,+\infty]$ and $1/p+1/q=1$,
\begin{align}
&L^p(\M,\omega)=h_\omega^{1\over2q}L^p(\M)h_\omega^{1\over2q}\ (\subseteq L^1(\M)),
\label{F-C.1}\\
&\|h_\omega^{1\over2q}ah_\omega^{1\over2q}\|_{p,\omega}=\|a\|_p,\qquad a\in L^p(\M),
\label{F-C.2}
\end{align}
that is, $L^p(\M)\cong L^p(\M,\omega)$ by the isometry
$a\mapsto h_\omega^{1\over2q}ah_\omega^{1\over2q}$. Interpolation $L^p$-spaces were
introduced in \cite{Ko1} in terms of more general embeddings
$x\in\M\mapsto h_\omega^\eta xh_\omega^{1-\eta}\in L^1(\M)$ with $0\le\eta\le1$. (The
$\eta=1/2$ case is the above symmetric $L^1(\M,\omega)$.) When $\M$ is general and
$\omega\in\M_*^+$ is not faithful with the support projection $e:=s(\omega)\in\M$, Kosaki's
$L^p$-space $L^p(\M,\omega)$ with respect to $\omega$ is still defined over $e\M e$ so that
\eqref{F-C.1} and \eqref{F-C.2} hold with $eL^p(\M)e$ in place of $L^p(\M)$.

Consider now the special case $\M=B(\cH)$, and let $\omega\in B(\cH)_*^+$ be given with
$e:=s(\omega)$ and $\hat\omega\in\cC_1(\cH)_+$ representing $\omega$. When $1\le p\le+\infty$
and $1/p+1/q=1$, Kosaki's $L^p$-space with respect to $\omega$ is
$L^p(B(\cH),\omega)=\hat\omega^{1\over2q}\cC_p(\cH)\hat\omega^{1\over2q}$ with
$\|\hat\omega^{1\over2q}A\hat\omega^{1\over2q}\|_{p,\omega}=\|A\|_p$ for $A\in e\cC_p(\cH)e$
(where $\cC_\infty(\cH)=B(\cH)$). In particular, when $\dim\cH<+\infty$,
$L^p(B(\cH),\omega)=eB(\cH)e=B(e\cH)$ and the interpolation $L^p$-norm is
$\|A\|_{p,\omega}=\|\hat\omega^{-{1\over2q}}A\hat\omega^{-{1\over2q}}\|_p$ for any
$A\in B(e\cH)$. The interpolation norm in the finite-dimensional case was used in \cite{Be} for
instance.

\section{Generalized conditional expectations}\label{Sec-D}

Let $\M$ and $\N$ be ($\sigma$-finite) von Neumann algebras, with standard forms
$(\M,\cH,J,\cP)$ and $(\N,\cH_0,J_0,\cP_0)$, respectively (see Appendix \ref{Sec-A}). Let
$\Phi:\N\to\M$ be a unital positive map. Let a faithful $\omega\in\M_*^+$ be given, and
assume that $\omega\circ\Phi$ is normal and faithful on $\N$.  In this case, $\Phi$ is automatically
normal and faithful (i.e., $\Phi(x^*x)=0$$\implies$$x=0$). Then it was shown in \cite{AC} that there
exists a unique unital normal positive map $\Phi_\omega^*:\M\to\N$ such that
\begin{align}\label{F-D.1}
\<Jx\Omega,\Phi(y)\Omega\>=\<J_0\Phi_\omega^*(x)\Omega_0,y\Omega_0\>,
\qquad x\in\M,\ y\in\N,
\end{align}
where $\Omega\in\cP$ and $\Omega_0\in\cP_0$ are the vector representatives of $\omega$
and $\omega\circ\Phi$, respectively. The map $\Phi_\omega^*$ is also faithful. Moreover,
we have
\begin{align}\label{F-D.2}
\omega\circ\Phi\circ\Phi_\omega^*=\omega,
\end{align}
and $\Phi_\omega^*$ is completely positive if and only if so is $\Phi$. This map
$\Phi_\omega^*$ is called the \emph{$\omega$-dual map} of $\Phi$, or the
\emph{Petz recovery map} (see \cite{Pe2}), whose definition by \eqref{F-D.1} is independent of
the choice of the standard forms of $\M,\N$. In terms of Haagerup's $L^1$-elements
$h_\omega$ and $h_{\omega\circ\Phi}$ (see Appendix \ref{Sec-B}), we note
\cite[Lemma 8.3]{Hi} that the map $\Phi_\omega^*$ is determined by
\begin{align}\label{F-D.3}
\Phi_*(h_\omega^{1/2}xh_\omega^{1/2})
=h_{\omega\circ\Phi}^{1/2}\Phi_\omega^*(x)h_{\omega\circ\Phi}^{1/2},
\qquad x\in\M,
\end{align}
where $\Phi_*:L^1(\M)\to L^1(\N)$ is the predual map of $\Phi$ via
$\M_*\cong L^1(\M)$ and $\N_*\cong L^1(\N)$, i.e.,
$\Phi_*(h_\psi)=h_{\psi\circ\Phi}$, $\psi\in\M_*$. Note that the construction of
$\Phi_\omega^*$ is possible even when $\omega$ and/or $\omega\circ\Phi$ are not
faithful (see \cite[Theorem 6.1 and Lemma 8.3]{Hi}), though the above setting is sufficient for our
present purpose.
 
In particular, let $\N$ be a von Neumann subalgebra of $\M$ containing the unit of $\M$, and
$\omega\in\M_*^+$ be faithful. The $\omega$-dual map $\Phi_\omega^*$ of the injection
$\Phi:\N\hookrightarrow\M$ is called the \emph{generalized conditional expectation} with
respect to $\omega$ \cite{AC}, which we denote by $\cE_{\N,\omega}:\M\to\N$. The map
$\cE_{\N,\omega}$ is unital, normal, completely positive, and faithful. Property \eqref{F-D.2}
becomes
\[
\omega\circ\cE_{\N,\omega}=\omega.
\]
In the present case, the standard Hilbert space $\cH_0$ for $\N$ is taken as
$\cH_0=\overline{\N\Omega}$, where the vector representative $\Omega_0$ of
$\omega\circ\Phi=\omega|_\N$ is equal to $\Omega$. Let $P$ be the orthogonal projection
from $\cH=\overline{\M\Omega}$ onto $\cH_0=\overline{\N\Omega}$. In this situation, note
\cite{AC} that $\cE_{\N,\omega}=\Phi_\omega^*$ given in \eqref{F-D.1} and \eqref{F-D.3} can
be written more explicitly as
\begin{align}\label{F-D.4}
\cE_{\N,\omega}(x)=J_0PJxJPJ_0=J_0PJxJJ_0,\qquad x\in\M,
\end{align}
which is also determined by $\cE_{\N,\omega}(x)\Omega=J_0PJx\Omega$, $x\in\M$.  As is
well known \cite{Ta0}, there exists a (genuine) \emph{conditional expectation} (i.e., a norm-one
projection) $E:\M\to\N$ such that $\omega\circ E=\omega$ on $\M$, if and only if $\N$ is
globally invariant under the modular automorphism group $\sigma_t^\omega$
(see Appendix \ref{Sec-B}) of $\M$ with respect to $\omega$, i.e., $\sigma_t^\omega(\N)=\N$,
$t\in\bR$. If this is the case, $J_0=J|_{\cH_0}$ and $JP=PJ$ hold so that
$\cE_{\N,\omega}=E$. An important property of $E$ is the bimodule property
$E(axb)=aE(x)b$ for $a,b\in\N$ and $x\in\M$, which $\cE_{\N,\omega}$ does not satisfy
in general. A merit of $\cE_{\N,\omega}$ is that it always exists, while the existence of $E$ is
very restrictive as stated above.

\section{Injective von Neumann algebras and nuclear $C^*$-algebras}\label{Sec-E}
\label{sec:injective}

A von Neumann algebra $\M$ on a Hilbert space $\cH$ is \emph{injective} if and only if there
exists a (not necessarily normal) conditional expectation (i.e., a projection of norm one
\cite{Tom}) from $\B(\cH)$ onto $\M$; see, e.g., \cite[Corollary XV.1.3]{Ta3}. A fundamental
result of Connes \cite{Co} (see also \cite[Theorem XVI.1.9]{Ta3}) says that a von Neumann
algebra $\M$ of separable predual is injective if and only if $\M$ is
\emph{AFD (approximately finite dimensional)}, i.e., there exists an increasing sequence
$\{\M_j\}_{j=1}^\infty$ of finite-dimensional *-subalgebras of $\M$ such that
$\M=\bigl(\bigcup_{j=1}^\infty\M_j\bigr)''$. In \cite{Ell} the result was furthermore extended in
such a way that a (general) von Neumann algebra $\M$ is injective if and only if there is an
increasing net $\{\M_i\}_{i\in\cI}$ of finite-dimensional *-subalgebras of $\M$ with
$\M=\bigl(\bigcup_{i\in\cI}\M_i\bigr)''$. (Here, $\cA''$ denotes the double commutant, i.e., the
commutant of $\cA'$, for any $\cA\subseteq\B(\hil)$.)

Next, a $C^*$-algebra $\cA$ is said to be \emph{nuclear} if, for every $C^*$-algebra $\B$,
there is a unique $C^*$-cross-norm on $\cA\odot\B$, i.e.,
$\cA\otimes_{\min}\B=\cA\otimes_{\max}\B$; see, e.g., \cite[Chap.~XV]{Ta3}. Concerning
nuclear $C^*$-algebras, among many others, the most fundamental result is that $\cA$ is nuclear
if and only if $\cA^{**}$ is injective. Here, $\cA^{*}$ denotes the Banach space dual of $\cA$,
and $\cA^{**}$ the second Banach space dual of $\cA$. Note that $\cA^{**}$ is isometrically
isomorphic to the \emph{universally enveloping von Neumann algebra} of $\cA$, and so it is
customary to use $\cA^{**}$ to denote the latter as well; see \cite[Sec.~III.2]{Ta1}. Therefore, if
$\cA$ is nuclear, then $\pi(\cA)''$ is injective for every representation $\pi$ of $\cA$. Typical
examples of nuclear $C^*$-algebras are AF $C^*$-algebras, in particular, the compact operator
ideal $\cC(\cH)$ (or rather $\cC(\cH)+\bC1$ in our present setting; see Example \ref{E-4.6}).
Here, recall that a $C^*$-algebra $\cA$ is AF if there exists an increasing sequence
$\{\cA_k\}_{k=1}^\infty$ of finite-dimensional *-subalgebras of $\cA$ such that
$\bigcup_{k=1}^\infty\cA_k$ is norm-dense in $\cA$. More intricate examples are provided by
groups. For a discrete group $G$, the $C^*$-algebra generated by the left regular representation
on $\ell^2(G)$ is the \emph{(reduced) group $C^*$-algebra} $C_r^*(G)$, while the generated
von Neumann algebra is the \emph{group von Neumann algebra} $W^*(G)$. Then $G$ is
amenable\,$\iff$\,$C_r^*(G)$ is nuclear\,$\iff$\,$W^*(G)$ is injective.

\section{Strong converse exponent in the finite-dimensional case}\label{Sec-F}
\label{sec:finitedim}

In this appendix we assume that a von Neumann algebra $\M$ is finite-dimensional, so
$\M\subseteq\B(\cH)$ with a finite-dimensional Hilbert space $\cH$. Note that $\M$ is isomorphic
to $\bigoplus_{i=1}^m\B(\cH_i)$, a finite direct sum of finite-dimensional $\B(\cH_i)$, so it is clear
that all the arguments in \cite{MO1} are valid with $\M$ in place of $\B(\cH)$. Let $\Tr$ be the
usual trace on $\M$ (such that $\Tr(e)=1$ of all minimal projections $e\in\M$). Below, to designate
states of $\M$, we use density operators $\rho,\sigma$ with respect to $\Tr$ rather than positive
functionals. Recall that both of the relative entropy $D(\rho\|\sigma)$ and the max-relative entropy
$D_{\max}(\rho\|\sigma)$ showing up in \eqref{F-2.4}--\eqref{F-2.6} play an important role to describe
$\psi(s):=\psi^*(\rho\|\sigma|s+1)$ and $H_r^*(\rho\|\sigma)$ in \cite[Sec.~4]{MO1}.

The aim of this appendix is to give Proposition \ref{P-F.2}, which is used in Section \ref{Sec-3.2}.
The main assertion is that for finite-dimensional density operators we have
$sc_r^0(\rho\|\sigma)\le H_r^*(\rho\|\sigma)$, $r\ge 0$, which in turn can be obtained easily from
the weaker inequalities $sc_r(\rho\|\sigma)\le H_r^*(\rho\|\sigma)$, $r\ge 0$. The latter was proved
in \cite{MO1}; however, the proof contains a gap, as it is implicitly assumed there that
$D(\rho\|\sigma)<D_{\max}(\rho\|\sigma)$. Our main contribution in Proposition \ref{P-F.2} is filling
this gap; for this we give a characterization of the case $D(\rho\|\sigma)=D_{\max}(\rho\|\sigma)$,
which may be of independent interest.

\begin{lemma}\label{L-F.1}
For density operators $\rho,\sigma$ in $\M$ with $s(\rho)\le s(\sigma)$, the following conditions
are equivalent:
\begin{itemize}
\item[\rm(a)] $s\mapsto\psi(s):=\psi^*(\rho\|\sigma|s+1)$ is affine on $(0,+\infty)$;
\item[\rm(b)] $D(\rho\|\sigma)=D_{\max}(\rho\|\sigma)$;
\item[\rm(c)] $\rho$ and $\sigma$ commute, and $\rho\sigma^{-1}=\gamma s(\rho)$ for some
constant $\gamma>0$;
\item[\rm(d)] $s(\rho)\sigma=\sigma s(\rho)$ and $\rho=\gamma\sigma s(\rho)$ for some constant
$\gamma>0$.
\end{itemize}
Moreover, if the above hold, then we have $\gamma\ge1$, $D(\rho\|\sigma)=\log\gamma$ and
\begin{align}\label{F-F.1}
H_r^*(\rho\|\sigma)=(r-D(\rho\|\sigma))_+,\qquad r\ge0.
\end{align}
\end{lemma}

\begin{proof}
(a)$\iff$(b).\enspace
Since $\psi(s)$ is a differentiable convex function on $[0,+\infty)$, this is clear from
\cite[Lemma 4.2]{MO1}.

(b)$\implies$(c).\enspace
Consider $D_2(\rho\|\sigma):=\log\Tr\rho^2\sigma^{-1}$, the standard (or Petz-type) R\'enyi
$2$-divergence of $\rho,\sigma$. Note that
\[
D(\rho\|\sigma)\le D_2^*(\rho\|\sigma)\le D_2(\rho\|\sigma)\le D_{\max}(\rho\|\sigma),
\]
where the first inequality is seen from the properties noted in Section \ref{Sec-2}, the second is
due to the Araki--Lieb--Thirring inequality, and the last was shown in \cite[Lemma 7]{BD}. Hence (b)
implies that $D_2^*(\rho\|\sigma)=D_2(\rho\|\sigma)$, i.e.,
$\Tr(\sigma^{-1/4}\rho\sigma^{-1/4})^2=\Tr\,\sigma^{-1/2}\rho^2\sigma^{-1/2}$. Using
\cite[Theorem 2.1]{Hi2} we find that $\rho,\sigma$ commute. Hence, (a) says that
$\psi(s)=\log\Tr\,\rho^{s+1}\sigma^{-s}$ is affine on $(0,+\infty)$. From \cite[Lemma 3.2]{HMO} for the
commutative case, it follows that $\rho\sigma^{-1}=\gamma s(\rho)$, implying (c).

(c)$\implies$(d) is obvious.

(d)$\implies$(b).\enspace
From condition (c) it easily follows that $\gamma\ge1$ and $D_{\max}(\rho\|\sigma)=\log\gamma$.
Moreover, since
\begin{align*}
D(\rho\|\sigma)&=\Tr(\rho\log\rho-\rho\log(s(\rho)\sigma)) \\
&=\Tr(\rho\log\rho-\rho\log(\gamma^{-1}\rho))=\log\gamma,
\end{align*}
(b) follows.

Finally, if (b) and hence (a) hold, then $\psi(s)=D(\rho\|\sigma)s$ for all $s>0$, from which \eqref{F-F.1}
follows immediately.
\end{proof}

The next proposition is used in the proof of Theorem \ref{T-3.7}, while the former is a specialized case
of the latter.

\begin{proposition}\label{P-F.2}
For every density operators $\rho,\sigma$ in $\M$ with $s(\rho)\le s(\sigma)$ and any $r\ge0$ we have
\[
\underline{sc}_r(\rho\|\sigma)=sc_r^0(\rho\|\sigma)=H_r^*(\rho\|\sigma).
\]
\end{proposition}

\begin{proof}
Since $sc_r^0(\rho\|\sigma)\ge\underline{sc}_r(\rho\|\sigma)\ge H_r^*(\rho\|\sigma)$, where the second
inequality is by \cite[Lemma 4.7]{MO1}, it suffices to prove that
$sc_r^0(\rho\|\sigma)\le H_r^*(\rho\|\sigma)$, $r\ge 0$.  Moreover, the last inequality follows
if we can prove that  $sc_r(\rho\|\sigma)\le H_r^*(\rho\|\sigma)$, $r\ge 0$, since then
\begin{align*}
sc_r^0(\rho\|\sigma)\le
\inf_{r'>r}sc_r(\rho\|\sigma)
\le
\inf_{r'>r} H_{r'}^*(\rho\|\sigma)
=
H_r^*(\rho\|\sigma),
\end{align*}
where the first inequality is obvious by definition, the second inequality is 
to be proved below, and the equality follows from the fact that $r\mapsto H_r^*(\rho\|\sigma)$ is a
monotone increasing finite-valued convex function on $\bR$, whence it is also continuous.

Let us therefore prove $sc_r(\rho\|\sigma)\le H_r^*(\rho\|\sigma)$, $r\ge 0$.  The proof of
\cite[Theorem 4.10]{MO1} gives this when $D(\rho\|\sigma)<D_{\max}(\rho\|\sigma)$.
Assume thus that $D(\rho\|\sigma)=D_{\max}(\rho\|\sigma)=:D$. 
For any $r\ge 0$, the test sequence $T_{n,r}:=e^{-n(r-D)_+}s(\rho)^{\otimes n}$, $n\in\bN$,
yields 
\begin{align*}
-\frac{1}{n}\log\Tr\,\rho_nT_{n,r}=(r-D)_+=H_r\nw(\rho\|\sigma),
\end{align*}
where the last equality is by \eqref{F-F.1}, and 
\[
-\frac{1}{n}\log\Tr\,\sigma_nT_{n,r}
=-\frac{1}{n}\log\bigl( e^{-n(r-D)_+}\Tr(\sigma s(\rho))^{\otimes n}\bigr)
=D+(r-D)_+\ge r,
\]
where we have used that $\sigma s(\rho)=e^{-D}\rho$ by Lemma \ref{L-F.1}.
This proves $sc_r(\rho\|\sigma)\le H_r\nw(\rho\|\sigma)$. 
\end{proof}

\section{Boundary values of convex functions on $(0,1)$}\label{Sec-G}
\label{sec:boundary}

Let $\{\phi_i\}_{i\in\cI}$ be a set of convex functions on $(0,1)$ with values in $(-\infty,+\infty]$. Define
\[
\phi(u):=\sup_{i\in\cI}\phi_i(u),\qquad u\in(0,1),
\]
which is obviously convex on $(0,1)$ with values in $(-\infty,+\infty]$. We extend $\phi_i$ and $\phi$
to $[0,1]$ by continuity as
\begin{align*}
\phi_i(u)&:=\lim_{u\searrow0}\phi_i(u),\qquad\phi_i(1):=\lim_{u\nearrow1}\phi_i(u), \\
\phi(0)&:=\lim_{u\searrow0}\phi(u),\qquad\ \ \phi(1):=\lim_{u\nearrow1}\phi(u).
\end{align*}
We then give the next lemma to use it in the proof of Theorem \ref{T-3.7}.

\begin{lemma}\label{L-G.1}
In the situation stated above, if $\phi(u)<+\infty$ for some $u\in(0,1)$, then
\[
\phi(0)=\sup_{i\in\cI}\phi_i(0),\qquad\phi(1)=\sup_{i\in\cI}\phi_i(1).
\]
\end{lemma}

\begin{proof}
By assumption we have a $u_0\in(0,1)$ with $\phi(u_0)<+\infty$. Obviously, $\phi_i(0)\le\phi(0)$ and
$\phi_i(1)\le\phi(1)$ for all $i\in\cI$. Hence it suffices to show that $\phi(0)\le\sup_i\phi_i(0)$ and
$\phi(1)\le\sup_i\phi_i(1)$. Let us prove the first inequality (the proof of the latter is similar). Set
$\xi:=\sup_i\phi_i(0)$. If $\xi=+\infty$, the assertion holds trivially. So assume $\xi<+\infty$. By
convexity, for every $i\in\cI$ we have
\[
\phi_i(u)\le{u_0-u\over u_0}\,\xi+{u\over u_0}\,\phi(u_0),\qquad u\in(0,u_0),
\]
so that
\[
\phi(u)\le{u_0-u\over u_0}\,\xi+{u\over u_0}\,\phi(u_0),\qquad u\in(0,u_0).
\]
This implies that $\phi(0)\le\xi=\sup_i\phi_i(0)$.
\end{proof}

\section{Proof of Theorem \ref{T-4.3}}\label{Sec-H}
\label{sec:proof}

Let $\rho,\sigma\in\cA_+^*$ and $\pi$ be any $(\rho,\sigma)$-normal representation of $\cA$
with $\tilde\rho=\rho_\pi$ and $\tilde\sigma=\sigma_\pi$, the normal extensions to
$\M:=\pi(\cA)''$. Also, let $\overline\rho,\overline\sigma$ be the normal extensions of
$\rho,\sigma$ to the enveloping von Neumann algebra $\cA^{**}$ and
$\overline\pi:\cA^{**}\to\M$ be the normal extension of $\pi$ to $\cA^{**}$ (see
\cite[p.~121]{Ta1}). Let $s(\overline\pi)$ be the support projection of $\overline\pi$.
Concerning the support projections $s(\tilde\rho)$ and $s(\overline\rho)$ we have
$s(\tilde\rho)=\overline\pi(s(\overline\rho))$ with $s(\overline\rho)\le s(\overline\pi)$. Therefore,
$s(\tilde\rho)\le s(\tilde\sigma)$ is equivalent to $s(\overline\rho)\le s(\overline\sigma)$. This
means that the condition $s(\tilde\rho)\le s(\tilde\sigma)$ is independent of the choice of a
representation $\pi$ as above. (The condition is called the \emph{absolute continuity} of $\rho$
with respect to $\sigma$ \cite{Hi5}.)

Now let $\hat\pi$ be another $(\rho,\sigma)$-normal representation of $\cA$ with
$\hat\rho:=\rho_{\hat\pi}$ and $\hat\sigma:=\sigma_{\hat\pi}$, the normal extensions to
$\hat\M:=\hat\pi(\cA)''$. The next lemma is a main ingredient of the proof of Theorem \ref{T-4.3}.

\begin{lemma}\label{L-H.1}
In the situation stated above, assume that $s(\tilde\rho)\le s(\tilde\sigma)$ (hence
$s(\hat\rho)\le s(\hat\sigma)$ as well). Let $z_0,\hat z_0$ denote the central supports of
$s(\tilde\sigma),s(\hat\sigma)$, respectively. Then there exists an isomorphism
$\Lambda:\M z_0\to\hat\M\hat z_0$ for which we have
\begin{align}\label{F-H.1}
\tilde\rho(x)=\hat\rho\circ\Lambda(x),\quad\tilde\sigma(x)=\hat\sigma\circ\Lambda(x),
\qquad x\in\M z_0,
\end{align}
and for every $p\in[1,+\infty)$,
\begin{align}\label{F-H.2}
\tr(h_{\tilde\sigma}^{1/2p}xh_{\tilde\sigma}^{1/2p})^p
=\tr(h_{\hat\sigma}^{1/2p}\Lambda(x)h_{\hat\sigma}^{1/2p})^p,\qquad x\in\M z_0,
\end{align}
where $h_{\tilde\sigma}\in L^1(\M)_+$ and $h_{\hat\sigma}\in L^1(\hat\M)_+$ are Haagerup's
$L^1$-elements corresponding to $\tilde\sigma\in\M_*^+$ and $\hat\sigma\in\hat\M_*^+$,
respectively.
\end{lemma}

\begin{proof}
We will work in the standard forms
\[
(\M,L^2(\M),J=\,^*,L^2(\M)_+),\qquad(\hat\M,L^2(\hat\M),\hat J=\,^*,L^2(\hat\M)_+).
\]
For brevity we write
\[
\begin{cases}
h_0:=h_{\tilde\rho}\in L^1(\M)_+, \\
k_0:=h_{\tilde\sigma}\in L^1(\M)_+, \\
e_0:=s(\tilde\sigma)=s(k_0)\in\M, \\
e_0':=Je_0J\in\M',
\end{cases}\qquad
\begin{cases}
\hat h_0:=h_{\hat\rho}\in L^1(\hat\M)_+, \\
\hat k_0:=h_{\hat\sigma}\in L^1(\hat\M)_+, \\
\hat e_0:=s(\hat\sigma)=s(\hat k_0)\in\hat\M, \\
\hat e_0':=\hat J\hat e_0\hat J\in\hat\M'.
\end{cases}
\]
Below the proof is divided into several steps.

{\it Step 1.}\enspace
Note that
\[
\overline{\pi(\cA)k_0^{1/2}}=\overline{\M k_0^{1/2}}=L^2(\M)e_0=e_0'L^2(\M)
\]
and for every $a\in\cA$,
\[
\<k_0^{1/2},\pi(a)e_0'k_0^{1/2}\>=\<k_0^{1/2},\pi(a)k_0^{1/2}\>
=\tilde\sigma\circ\pi(a)=\sigma(a).
\]
Hence $\{\pi(\cdot)e_0',e_0'L^2(\M),k_0^{1/2}\}$ is the cyclic representation of $\cA$ with respect to
$\sigma$, and similarly $\{\hat\pi(\cdot)\hat e_0',\hat e_0'L^2(\hat\M),\hat k_0^{1/2}\}$ is the same.
By the uniqueness (up to unitary conjugation) of the cyclic representation, there exists a unitary
$V:L^2(\M)e_0\to L^2(\hat\M)\hat e_0$ such that
\begin{align}\label{F-H.3}
Vk_0^{1/2}=\hat k_0^{1/2},\qquad
V(\pi(a)e_0')V^*=\hat\pi(a)\hat e_0',\quad a\in\cA.
\end{align}
We hence have an isomorphism $V\cdot V^*:\M e_0'\to\hat\M\hat e_0'$.

{\it Step 2.}\enspace
Since $z_0$ is the central support of $e_0'$, note that $xz_0\in\M z_0\mapsto xe_0'\in\M e_0'$
($x\in\M$) is an isomorphism, and similarly so is
$\hat x\hat z_0\in\hat\M\hat z_0\mapsto\hat x\hat e_0'\in\hat \M\hat e_0'$ ($\hat x\in\hat\M$).
Hence one can define an isomorphism $\Lambda:\M z_0\to\hat\M\hat z_0$ as follows:
\begin{align}\label{F-H.4}
\Lambda:\,\M z_0\,\cong\,\M e_0'\,\cong\,\hat\M\hat e_0'
\,\cong\,\hat\M\hat z_0,\qquad
xz_0\,\mapsto xe_0'\,\mapsto\,V(xe_0')V^*=\hat x\hat e_0'\,\mapsto\,\hat x\hat z_0.
\end{align}
Note \cite[Lemma 2.6]{Ha} that the standard forms of $\M z_0$ and $\hat\M\hat z_0$ are
respectively given by
\begin{align*}
&(\M z_0,\,z_0L^2(\M)z_0=L^2(\M)z_0,\,J=\,^*,\,z_0L^2(\M)_+z_0=L^2(\M)_+z_0), \\
&(\hat\M\hat z_0,\,\hat z_0L^2(\hat\M)\hat z_0=L^2(\hat\M)\hat z_0,\,\hat J=\,^*,\,
\hat z_0L^2(\hat\M)_+\hat z_0=L^2(\hat\M)_+\hat z_0).
\end{align*}
By the uniqueness (up to unitary conjugation) of the standard form (under isomorphism)
\cite[Theorem 2.3]{Ha}, there exists a unitary $U:L^2(\M)z_0\to L^2(\hat\M)\hat z_0$ such that
\begin{align}
&\Lambda(x)=UxU^*,\qquad x\in\M z_0, \label{F-H.5}\\
&(U\xi)^*=U(\xi^*),\qquad\xi\in z_0L^2(\M)z_0. \label{F-H.6}\\
&U(L^2(\M)_+z_0)=L^2(\hat\M)_+\hat z_0. \label{F-H.7}
\end{align}

{\it Step 3.}\enspace
Since $s(h_0)\le e_0\le z_0$ by assumption, one has $h_0^{1/2},k_0^{1/2}\in L^2(\M)z_0$, and
similarly $\hat h_0^{1/2},\hat k_0^{1/2}\in L^2(\hat\M)\hat z_0$. By \eqref{F-H.7} one has
$Uh_0^{1/2},Uk_0^{1/2}\in L^2(\hat\M)_+\hat z_0$. Here we confirm that
\begin{align}\label{F-H.8}
Uh_0^{1/2}=\hat h_0^{1/2},\qquad Uk_0^{1/2}=\hat k_0^{1/2}.
\end{align}
To show this, for every $a\in\cA$ we find that
\begin{align*}
\<Uh_0^{1/2},(\hat\pi(a)\hat z_0)Uh_0^{1/2}\>
&=\<h_0^{1/2},\Lambda^{-1}(\hat\pi(a)\hat z_0)h_0^{1/2}\>\quad\mbox{(by \eqref{F-H.5})} \\
&=\<h_0^{1/2},(\pi(a)z_0)h_0^{1/2}\>\quad\mbox{(by \eqref{F-H.3} and \eqref{F-H.4}}) \\
&=\tilde\rho\circ\pi(a)=\rho(a)=\hat\rho\circ\hat\pi(a) \\
&=\<\hat h_0^{1/2},(\hat\pi(a)\hat z_0)\hat h_0^{1/2}\>,
\end{align*}
which implies that $Uh_0^{1/2}=\hat h_0^{1/2}$. The proof of $Uk_0^{1/2}=\hat k_0^{1/2}$ is similar.
By \eqref{F-H.5} and \eqref{F-H.8} we have also
\begin{align}\label{F-H.9}
\Lambda(x)\hat h_0^{1/2}=U(xh_0^{1/2}),\quad\Lambda(x)\hat k_0^{1/2}=U(xk_0^{1/2}),
\qquad x\in\M z_0.
\end{align}
These imply \eqref{F-H.1}. Furthermore, by \eqref{F-H.8} and \eqref{F-H.9} we have
$\Lambda(e_0)\hat k_0^{1/2}=Uk_0^{1/2}=\hat k_0^{1/2}$, from which $\Lambda(e_0)\ge\hat e_0$
follows. Applying the same argument to $\Lambda^{-1}(\hat x)=U^*\hat xU$
($\hat x\in\hat\M\hat z_0$) with $k_0^{1/2},\hat k_0^{1/2}$ exchanged gives
$\Lambda^{-1}(\hat e_0)\ge e_0$ as well. Therefore,
\begin{align}\label{F-H.10}
\Lambda(e_0)=\hat e_0.
\end{align}

{\it Step 4.}\enspace
We define
\[
(\Lambda^{-1})_*:L^1(\M z_0)=L^1(\M)z_0\,\to\,L^1(\hat\M\hat z_0)=L^1(\hat\M)\hat z_0
\]
by transforming $\psi\in(\M z_0)_*\mapsto\psi\circ\Lambda^{-1}\in(\hat\M\hat z_0)_*$ via
$L^1(\M z_0)\cong(\M z_0)_*$ and $L^1(\hat\M\hat z_0)\cong(\hat\M\hat z_0)_*$, that is,
$(\Lambda^{-1})_*:h_\psi\in L^1(\M z_0)\mapsto
h_{\psi\circ\Lambda^{-1}}\in L^1(\hat\M\hat z_0)$ for $\psi\in(\M z_0)_*$. Of course,
$(\Lambda^{-1})_*$ is an isometry with respect to $\norm{\cdot}_1$. Now, Kosaki's (symmetric)
interpolation $L^p$-spaces enter into our discussions. Here we confirm that
\begin{align}\label{F-H.11}
(\Lambda^{-1})_*(k_0^{1/2}xk_0^{1/2})=\hat k_0^{1/2}\Lambda(x)\hat k_0^{1/2},
\qquad x\in\M z_0.
\end{align}
Indeed, for every $x,y\in\M z_0$ we find that
\begin{align*}
\tr\bigl(\Lambda(y)(\Lambda^{-1})_*(k_0^{1/2}xk_0^{1/2})\bigr)
&=\tr(yk_0^{1/2}xk_0^{1/2})=\<(yk_0^{1/2})^*,xk_0^{1/2}\> \\
&=\<U((yk_0^{1/2})^*),U(xk_0^{1/2})\> \\
&=\<(\Lambda(y)\hat k_0^{1/2})^*,\Lambda(x)\hat k_0^{1/2}\>\quad
\mbox{(by \eqref{F-H.6} and \eqref{F-H.9})} \\
&=\tr\bigl(\Lambda(y)\hat k_0^{1/2}\Lambda(x)\hat k_0^{1/2}\bigr),
\end{align*}
which yields \eqref{F-H.11}.

{\it Step 5.}\enspace
Thanks to \eqref{F-H.11} we see that the isometry
$(\Lambda^{-1})_*:L^1(\M z_0)\to L^1(\hat\M\hat z_0)$ with respect to $\norm{\cdot}_1$ is
restricted to an isometry from $k_0^{1/2}(\M z_0)k_0^{1/2}$ (embedded into $L^1(\M z_0)$) onto
$\hat k_0^{1/2}(\hat\M\hat z_0)\hat k_0^{1/2}$ (embedded into $L^1(\hat\M\hat z_0)$) with
respect to $\norm{\cdot}_\infty$, i.e.,
\[
\|k_0^{1/2}xk_0^{1/2}\|_\infty=\|x\|=\|\Lambda(x)\|
=\|\hat k_0^{1/2}\Lambda(x)\hat k_0^{1/2}\|_\infty,\qquad x\in\M z_0.
\]
By Kosaki's construction in \cite{Ko1} (or the Riesz--Thorin theorem) it follows that
$(\Lambda^{-1})_*$ gives rise to an isometry
\begin{align*}
(\Lambda^{-1})_*:&\,L^p(\M z_0,\tilde\sigma)
=C_{1/p}((k_0^{1/2}(\M z_0)k_0^{1/2},L^1(\M z_0)) \\
&\qquad\to\,L^p(\hat\M\hat z_0,\hat\sigma)
=C_{1/p}(\hat k_0^{1/2}(\hat\M\hat z_0)\hat k_0^{1/2},L^1(\hat\M\hat z_0))
\end{align*}
with respect to the interpolation norms $\norm{\cdot}_{p,\tilde\sigma}$ and
$\norm{\cdot}_{p,\hat\sigma}$ for any $p\in[1,+\infty)$. For every $x\in\M z_0$, applying this to
$k_0^{1/2}xk_0^{1/2}$ with \eqref{F-H.11} gives
\[
\|k_0^{1/2}xk_0^{1/2}\|_{p,\tilde\sigma}=\|\hat k_0^{1/2}\Lambda(x)\hat k_0^{1/2}\|_{p,\hat\sigma}.
\]
By \cite[Theorem 9.1]{Ko1}, for every $p\in[1,+\infty)$ the above equality is rephrased as Haagerup's
$L^p$-norm equality
\[
\|k_0^{1/2p}xk_0^{1/2p}\|_p=\|\hat k_0^{1/2p}\Lambda(x)\hat k_0^{1/2p}\|_p,
\]
which is \eqref{F-H.2}, as asserted.
\end{proof}

We are now in a position to prove Theorem \ref{T-4.3}.

\begin{proof}[Proof of (i)]
We use the variational expressions in Proposition \ref{P-2.3} based on Lemma \ref{L-H.1}. Assume
first that $\alpha>1$. If
$s(\tilde\rho)\not\le s(\tilde\sigma)$, then $s(\hat\rho)\not\le s(\hat\sigma)$ (as mentioned at the
beginning of this appendix) so that both of $D_\alpha^*(\tilde\rho\|\tilde\sigma)$ and
$D_\alpha^*(\hat\rho\|\hat\sigma)$ are $+\infty$. Hence we assume that
$s(\tilde\rho)\le s(\tilde\sigma)$ (and $s(\hat\rho)\le s(\hat\sigma)$). Using \eqref{F-2.7} we have
\begin{align*}
Q_\alpha^*(\tilde\rho\|\tilde\sigma)
&=\sup_{x\in\M_+}\Bigl[\alpha\tilde\rho(x)
-(\alpha-1)\tr\bigl(h_{\tilde\sigma}^{\alpha-1\over2\alpha}xh_{\tilde
\sigma}^{\alpha-1\over2\alpha}\bigr)^{\alpha\over\alpha-1}\Bigr] \\
&=\sup_{x\in(\M z_0)_+}\Bigl[\alpha\tilde\rho(x)-(\alpha-1)
\tr\bigl(h_{\tilde\sigma}^{\alpha-1\over2\alpha}xh_{\tilde\sigma}^{\alpha-1\over2\alpha}
\bigr)^{\alpha\over\alpha-1}\Bigr]\quad\mbox{(since $s(\tilde\rho)\le s(\tilde\sigma)\le z_0$)}\\
&=\sup_{x\in(\M z_0)_+}\Bigl[\alpha\hat\rho(\Lambda(x))-(\alpha-1)
\tr\bigl(h_{\hat\sigma}^{\alpha-1\over2\alpha}\Lambda(x)h_{\hat\sigma}^{\alpha-1\over2\alpha}
\bigr)^{\alpha\over\alpha-1}\Bigr]\quad\mbox{(by Lemma \ref{L-H.1})}\\
&=\sup_{\hat x\in(\hat\M\hat z_0)_+}\Bigl[\alpha\hat\rho(\hat x)-(\alpha-1)
\tr\bigl(h_{\hat\sigma}^{\alpha-1\over2\alpha}\hat xh_{\hat\sigma}^{\alpha-1\over2\alpha}
\bigr)^{\alpha\over\alpha-1}\Bigr] \\
&=Q_\alpha^*(\hat\rho\|\hat\sigma).
\end{align*}

Next, assume that $1/2\le\alpha<1$. When $s(\tilde\rho)\le s(\tilde\sigma)$, we use \eqref{F-2.8}
as above to have
\begin{align*}
Q_\alpha^*(\tilde\rho\|\tilde\sigma)
&=\inf_{x\in(\M z_0)_{++}}\Bigl[\alpha\tilde\rho(x)+(1-\alpha)
\tr\bigl(h_{\tilde\sigma}^{1-\alpha\over2\alpha}x^{-1}h_{\tilde\sigma}^{1-\alpha\over2\alpha}
\bigr)^{\alpha\over1-\alpha}\Bigr] \\
&=\inf_{x\in(\M z_0)_{++}}\Bigl[\alpha\hat\rho(\Lambda(x))+(1-\alpha)
\tr\bigl(h_{\hat\sigma}^{1-\alpha\over2\alpha}\Lambda(x^{-1})h_{\hat\sigma}^{1-\alpha\over2\alpha}
\bigr)^{\alpha\over1-\alpha}\Bigr]\quad\mbox{(by Lemma \ref{L-H.1})}\\
&=\inf_{\hat x\in(\hat\M\hat z_0)_{++}}\Bigl[\alpha\hat\rho(\hat x)+(1-\alpha)
\tr\bigl(h_{\hat\sigma}^{1-\alpha\over2\alpha}\hat x^{-1}h_{\hat\sigma}^{1-\alpha\over2\alpha}
\bigr)^{\alpha\over1-\alpha}\Bigr]\quad\mbox{(since $\Lambda(x^{-1})=\Lambda(x)^{-1}$)} \\
&=Q_\alpha^*(\hat\rho\|\hat\sigma).
\end{align*}
For general $\rho,\sigma\in\cA_+^*$, let $\sigma_\eps:=\sigma+\eps\rho$ for any $\eps>0$. Then
$\sigma_\eps$ has the normal extensions $\tilde\sigma_\eps=\tilde\sigma+\eps\tilde\rho$ to
$\M$ and $\hat\sigma_\eps=\hat\sigma+\eps\hat\rho$ to $\hat\M$. The above case yields
$Q_\alpha^*(\tilde\rho\|\tilde\sigma_\eps)=Q_\alpha^*(\hat\rho\|\hat\sigma_\eps)$ for all $\eps>0$.
From the continuity of $Q_\alpha^*$ on $\M_*^+\times\M_*^+$ in the norm topology when
$1/2\le\alpha<1$ (see \cite[Theorem 3.16\,(3)]{Hi}), letting $\eps\searrow0$ gives
$Q_\alpha^*(\tilde\rho\|\tilde\sigma)=Q_\alpha^*(\hat\rho\|\hat\sigma)$, implying \eqref{F-4.4}.
\end{proof}

\begin{proof}[Proof of (ii)]
Assume first that $s(\tilde\rho)\le s(\tilde\sigma)$ (hence $s(\hat\rho)\le s(\hat\sigma)$). Below
let us use the same symbols as in the proof of Lemma \ref{L-H.1}. Recall \cite{Ar2} that the relative
modular operator $\Delta_{\tilde\rho,\tilde\sigma}$ is defined as
$\Delta_{\tilde\rho,\tilde\sigma}:=S_{\tilde\rho,\tilde\sigma}^*\overline{S_{\tilde\rho,\tilde\sigma}}$,
where $S_{\tilde\rho,\tilde\sigma}$ is a closable conjugate linear operator defined by
\[
S_{\tilde\rho,\tilde\sigma}(xk_0^{1/2}+\zeta):=e_0x^*h_0^{1/2},\qquad
x\in\M,\ \zeta\in(L^2(\M)e_0)^\perp.
\]
Similarly,
$\Delta_{\hat\rho,\hat\sigma}:=S_{\hat\rho,\hat\sigma}^*\overline{S_{\hat\rho,\hat\sigma}}$ is
given, where
\[
S_{\hat\rho,\hat\sigma}(\hat x\hat k_0^{1/2}+\hat\zeta):=\hat e_0\hat x^*\hat h_0^{1/2},\qquad
\hat x\in\hat\M,\ \hat\zeta\in(L^2(\hat\M)\hat e_0)^\perp.
\]
Since $s(h_0)\le e_0\le z_0$, we can consider $S_{\tilde\rho,\tilde\sigma}$ and
$\Delta_{\tilde\rho,\tilde\sigma}$ as operators on $L^2(\M)z_0$ (they are zero operators on
$(L^2(\M)z_0)^\perp$). Similarly, $S_{\hat\rho,\hat\sigma}$ and $\Delta_{\hat\rho,\hat\sigma}$
are considered on $L^2(\hat\M)\hat z_0$. Let us use an isomorphism
$\Lambda:\M z_0\to\hat\M\hat z_0$ and a unitary $U:L^2(\M)z_0\to L^2(\hat\M)\hat z_0$.
Since $\overline{\M k_0^{1/2}}=L^2(\M)e_0$ and
$\overline{\hat\M\hat k_0^{1/2}}=L^2(\hat\M)\hat e_0$, it follows from \eqref{F-H.9} that
$U(L^2(\M)e_0)=L^2(\hat\M)\hat e_0$ and hence
$U((L^2(\M)e_0)^\perp)=(L^2(\hat\M)\hat e_0)^\perp$. For every
$\hat x=\Lambda(x)\in\hat\M\hat z_0$ (with $x\in\M z_0$) and
$\hat\zeta\in(L^2(\hat\M)\hat e_0)^\perp$, we find that
\begin{align*}
S_{\hat\rho,\hat\sigma}(\hat x\hat k_0^{1/2}+\hat\zeta)
&=Ue_0U^*Ux^*U^*Uh_0^{1/2}\quad
\mbox{(by \eqref{F-H.10}, \eqref{F-H.5} and \eqref{F-H.8})} \\
&=Ue_0x^*h_0^{1/2}=US_{\tilde\rho,\tilde\sigma}(xk_0^{1/2}+U^*\hat\zeta) \\
&=US_{\tilde\rho,\tilde\sigma}(U^*\hat x\hat k_0^{1/2}+U^*\hat\zeta)\quad
\mbox{(by \eqref{F-H.9})} \\
&=US_{\tilde\rho,\tilde\sigma}U^*(\hat x\hat k_0^{1/2}+\hat\zeta).
\end{align*}
This implies that $S_{\hat\rho,\hat\sigma}=US_{\tilde\rho,\tilde\sigma}U^*$ and hence
$\Delta_{\hat\rho,\hat\sigma}=U\Delta_{\tilde\rho,\tilde\sigma}U^*$. Therefore, for every
$\alpha\in[0,+\infty)$, $\hat k_0^{1/2}$ is in $\cD(\Delta_{\hat\rho,\hat\sigma}^{\alpha/2})$ if and
only if $k_0^{1/2}=U^*\hat k_0^{1/2}$ is in $\cD(\Delta_{\tilde\rho,\tilde\sigma}^{\alpha/2})$, and
in this case,
\[
Q_\alpha(\tilde\rho\|\tilde\sigma)=\|\Delta_{\tilde\rho,\tilde\sigma}^{\alpha/2}k_0^{1/2}\|^2
=\|\Delta_{\hat\rho,\hat\sigma}^{\alpha/2}\hat k_0^{1/2}\|^2=Q_\alpha(\hat\rho\|\hat\sigma).
\]
Otherwise, $Q_\alpha(\tilde\rho\|\tilde\sigma)=Q_\alpha(\hat\rho\|\hat\sigma)=+\infty$.

Next assume that $s(\tilde\rho)\not\le s(\tilde\sigma)$, equivalently
$s(\hat\rho)\not\le s(\hat\sigma)$. Then
$Q_\alpha(\tilde\rho\|\tilde\sigma)=Q_\alpha(\hat\rho\|\hat\sigma)=+\infty$ for $\alpha>1$. When
$0\le\alpha<1$, let $\sigma_\eps:=\sigma+\eps\rho$ for every $\eps>0$. From the continuity of
$Q_\alpha$ on $\M_*^+\times\M_*^+$ \cite[Corollary 3.8]{Hi}, the above case yields
\[
Q_\alpha(\tilde\rho\|\tilde\sigma)=\lim_{\eps\searrow0}Q_\alpha(\tilde\rho\|\tilde\sigma_\eps)
=\lim_{\eps\searrow0}Q_\alpha(\hat\rho\|\hat\sigma_\eps)=Q_\alpha(\hat\rho\|\hat\sigma),
\]
implying \eqref{F-4.5}.
\end{proof}

\begin{remark}\label{R-H.2}\rm
The notion of \emph{standard $f$-divergences} $S_f(\rho\|\sigma)$ with a parametrization of
operator convex functions $f$ on $(0,+\infty)$ has been studied in \cite{Hi3} in the von Neumann
algebra setting. From the above proof of (ii) we observe that $S_f(\rho\|\sigma)$ can be extended
to $\rho,\sigma\in\cA_+^*$ as
\[
S_f(\rho\|\sigma):=S_f(\rho_\pi\|\sigma_\pi)
\]
independently of the choice of a $(\rho,\sigma)$-normal representation $\pi$ of $\cA$. Then we
can easily extend properties of $S_f(\rho\|\sigma)$ given in \cite{Hi3} to the $C^*$-algebra setting
(like Proposition \ref{P-4.5} for the sandwiched and the standard R\'enyi divergences).
\end{remark}

\end{document}